\documentclass[journal]{IEEEtran}

\usepackage{color}
\usepackage{yfonts}
\usepackage{cite}
\usepackage[pdftex]{graphicx}
\setkeys{Gin}{clip=true,draft=false}
\DeclareGraphicsExtensions{.pdf}
\usepackage[cmex10]{amsmath}
\usepackage{amsthm}
\usepackage{amsfonts}
\usepackage{amssymb}
\usepackage[cspex,bbgreekl]{mathbbol}
\usepackage[hang]{subfigure}
\usepackage{color}
\usepackage{float}
\usepackage[super]{nth}
\usepackage{eufrak}
\usepackage{multirow} 
\newcommand{\xvec}{\mathbf{x}}
\newcommand{\yvec}{\mathbf{y}}

\newcommand{\bvec}{\mathbf{b}}
\newcommand{\pvec}{\mathbf{p}}
\newcommand{\kvec}{\mathbf{k}}

\newcommand{\svec}{\mathbf{s}}

\newcommand{\Zbb}{\mathbb{Z}_2}

\newcommand{\Rbb}{\mathbb{R}}
\newcommand{\Cbb}{\mathbb{C}}
\newcommand{\Nbb}{\mathbb{N}}
\newcommand{\T}{\mathcal{T}}
\newcommand{\I}{\mathcal{I}}
\newcommand{\nter}{\mathfrak{n}}
\newcommand{\mter}{\mathfrak{m}}
\newcommand{\poly}{\mathfrak{p}}

\newtheoremstyle{definition_new}% name of the style to be used
{0.7\topsep}% measure of space to leave above the theorem. E.g.: 3pt
{0.7\topsep}% measure of space to leave below the theorem. E.g.: 3pt
{\normalfont}% name of font to use in the body of the theorem
{0pt}% measure of space to indent
{\itshape}% name of head font
{}% punctuation between head and body
{7pt}% space after theorem head; " " = normal interword space
{}% Manually specify head
\theoremstyle{definition_new}
\newtheorem{defn}{Definition}[section]

\newtheorem{thm}{Theorem}[section]
\newtheorem{lem}[thm]{Lemma}
\newtheorem{prop}[thm]{Proposition}

\theoremstyle{remark}

\begin{document}

\title{Polynomial-time solution of prime factorization and NP-hard problems with digital memcomputing machines}

%\author{Fabio L. Traversa}
%\email[]{fabio.traversa@polito.it}
%\affiliation{Department of Physics, University of California, San Diego, La Jolla, CA 92093-0319, USA}

%\author{Massimiliano Di Ventra}
%\email[]{diventra@physics.ucsd.edu}
%\affiliation{Department of Physics, University of California, San Diego, La Jolla, CA 92093-0319, USA}

\author{Fabio L. Traversa, Massimiliano Di Ventra\thanks{The authors are with the Department of Physics, University of California-San Diego, 9500  Gilman Drive, La Jolla, California 92093-0319, USA, e-mail: ftraversa@physics.ucsd.edu, diventra@physics.ucsd.edu}}

\maketitle

\date{\today}
\begin{abstract}
We introduce a class of digital machines we name Digital Memcomputing Machines (DMMs) able to solve a wide range of problems including Non-deterministic Polynomial (NP) ones with polynomial resources (in time, space and energy). An abstract DMM with this power must satisfy a set of compatible mathematical constraints underlying its practical realization.
We initially prove this by introducing the complexity classes for these machines. We then make a connection with dynamical systems theory. 
This leads to the set of physical constraints for poly-resource resolvability. 
Once the mathematical requirements have been assessed, we propose a practical scheme to solve the above class of problems based on the novel concept of self-organizing logic gates and circuits (SOLCs). These are logic gates and circuits able to accept input signals from {\it any} terminal, without distinction between conventional input and output terminals. They can solve boolean problems by self-organizing into their solution. They can be fabricated either with circuit elements with memory (such as memristors) and/or standard MOS technology. Using tools of functional analysis, we prove mathematically the following constraints for the poly-resource resolvability: {\it i)} SOLCs possess a global attractor; {\it ii)} their only equilibrium points are the solutions of the problems to solve; {\it iii)} the system converges exponentially fast to
the solutions; {\it iv)} the equilibrium convergence rate scales at
most polynomially with input size. We finally provide arguments that periodic orbits and strange
attractors cannot coexist with equilibria. As examples we show how to solve the prime factorization and the NP-hard version of the subset-sum problem. Since DMMs map integers into integers they are robust against noise, and hence {\it scalable}. We finally discuss the implications of the DMM realization through SOLCs to the NP=P question related to constraints of poly-resources resolvability.
\end{abstract}


% Uncomment for PACS numbers title message
% \pacs{}
% Keywords required only for MST, PB, PMB, PM, JOA, JOB?
% \vspace{2pc}
% \noindent{\it Keywords}: Article preparation, IOP journals
% Uncomment for Submitted to journal title message
% \submitto{\JPA}
% Comment out if separate title page not required

%%%%%%%%%%%%%%%%%%%%
\section{Introduction}\label{introduction}
%%%%%%%%%%%%%%%%%%%%
We have recently shown that a new class of non-Turing machines, we have named Universal Memcomputing Machines (UMMs), has the same computational power of non-deterministic Turing machines, and as such they can solve (if implemented in hardware) NP-complete/hard problems with polynomial resources~\cite{UMM}. UMMs are machines composed 
of interacting memprocessors, namely processors that use memory to both process and store information on the same physical location~\cite{diventra13a}. Their computational 
power rests on their intrinsic parallelism - in that the {\it collective} state of the machine (rather than the individual memprocessors) computes - and their information overhead: a type of information that is embedded in the machine but not necessarily stored by it~\cite{UMM}. 

The information overhead is the direct result of the {\it topology} of the 
network of memprocessors, and can be appropriately ``tuned''  to the complexity of the problem. For instance, problems that would otherwise require exponential resources with 
a standard deterministic Turing machine, require only polynomial resources with a UMM if we feed into it an exponential information overhead through the appropriate topology. 

We have already provided a practical example of such a machine that solves the NP-complete version of the subset-sum problem in one computational step using a linear number of memprocessors and built it using standard electronic components~\cite{traversaNP}. However, the machine we have built is fully analog and as such can not be easily scaled up to an arbitrary number of memprocessors without noise control. However, UMMs can also be defined in digital mode and if a practical scheme can be engineered to realize them in practice, we could build scalable machines 
that solve very complex problems using resources that only grow polynomially with input size.

In this paper we suggest such a scheme and apply it to two important problems: prime factorization and the NP-hard version of the subset-sum problem. However, we accomplish more than this. The list of major results of our paper, with the corresponding sections where the reader can find them, are as follows.
\begin{itemize}
\item
We define the types of problems that we are interested in and the computational protocols we have in mind which will be important in practical realizations. Section~\ref{def_problem}. 
\item
We define digital memcomputing machines (DMMs) and their complexity classes. Section~\ref{DMM}.
\item
We formulate DMMs in terms of the theory of dynamical systems in order to make the transition from mathematical concept to physical systems and hence facilitate their practical realization. This leads to the emergence of a set of mathematical constraints that a dynamical system must satisfy in order to be able to solve certain problem classes with polynomial resources. Section~\ref{Dyna_section}.
\item 
We introduce the notion of {\it self-organizing logic gates} and {\it self-organizing logic circuits}, namely gates and circuits that can accept input signals from any terminal without 
differentiating between conventional input and output terminals. Through the collective state of the machine, they can satisfy boolean functions in a self-organized manner. Section~\ref{SOLG_section}. 
\item
Using tools of functional analysis, we demonstrate that the dynamical systems describing these circuits are dissipative (in the functional sense~\cite{hale_2010_asymptotic}, not necessarily in the physical passivity sense) and hence support a global attractor~\cite{hale_2010_asymptotic}. In addition, we prove that {\it i)} their equilibria are the solutions of the problems they represent, {\it ii)} any initial condition converges exponentially fast to the solution(s), {\it iii)} the equilibrium convergence rate scales at most polynomially with input size, {\it iv)} the energy expenditure only grows polynomially with input size, and {\it v)} we provide arguments that periodic orbits and strange attractors do not coexist with equilibria.  Section~\ref{Math_SOLC_section}.
\item
We support the above results with numerical simulations of the solution of prime factorization and the subset-sum problem using DMMs with appropriate topologies. The former problem scales as $O(n^2)$ in space (i.e., with the number of self-organizing logic gates employed) and $O(n^2)$ in convergence time with 
input size $n$. The latter as $O[p(n+\log_2(n-1))]$ in space and $O((n+p)^2)$ in convergence time with size $n$ and precision $p$. Section~\ref{NP_sec}. 
\item
We discuss the consequences of our results on the question of whether NP=P. Section~\ref{NP_P}.
\end{itemize}

Finally, Section~\ref{Conclusions} collects our thoughts for future directions in the field. 

\section{Preliminary Definitions}\label{def_problem}
Before defining the DMMs, we introduce the general class of problems we want to tackle. This will clarify the type of approach we will use to solve them, and will lead us to the 
definition of the machines' computational complexity classes.

%%%%%%%%%%%%%%%%%%%%
\subsection{Compact Boolean Problems}\label{Problem_subsection}
%%%%%%%%%%%%%%%%%%%%

\begin{defn}
A \textit{Compact Boolean} problem is a collection of statements (not necessarily independent) that can be written as a finite system of boolean functions. Formally, let $f:\Zbb^n\rightarrow \Zbb^m$ be a system of boolean functions and $\Zbb=\{0,1\}$, then a CB problem requires to find a solution $\yvec\in\Zbb^n$ (if it exists) of $f(\yvec)=\bvec$ with $\bvec\in\Zbb^m$.
\end{defn}

It is clear from this definition that this class includes many important problems in Computer Science, such as the Nondeterministic Polynomial- (NP-) time problems, linear algebra problems, and many others. For instance, such a class includes those we solve in Sec.~\ref{NP_sec}, namely factorization and the subset-sum problem. 

%%%%%%%%%%%%%%%%%%%%%
\subsection{Direct and Inverse Protocols}\label{Problem_subsection_meth}
%%%%%%%%%%%%%%%%%%%%%

We further define two ways (protocols) to find a solution of such problems. The first one can be implemented within the Turing machine paradigm through standard boolean circuits. The second one can only be implemented with DMMs.
\begin{defn}
Let $S_B=\{ g_1,\ldots,g_k \}$ be a set of $k$ boolean functions $g_j:\Zbb^{n_j}\rightarrow \Zbb^{m_j}$, $S_{CFS}=\{ s_1,\ldots,s_h\}$ a set of $h$ control flow statements, and $CF$ the control flow that specifies the sequence of the functions $g_j\in S_B$ to be evaluated and the statements $s_j\in S_{CFS}$ to be executed. We then define the \textit{Direct Protocol} (DP) for solving a CB problem that control flow $CF$ which takes $\yvec\rq\in\Zbb^{n\rq}$ as initial input and gives $\yvec\in\Zbb^n$ as final output such that $f(\yvec)=\bvec$. 
\end{defn}

Roughly speaking, the DP is the {\it ordered} set of instructions (the program) that a Turing machine should perform to solve a CB problem. Even if a DP is not the unique way to find a solution of a given CB problem with Turing machines, all the other known strategies (even if they are, possibly, more optimal) do not change the computational complexity classes of Turing machines. For this reason, we consider here only DPs as some general Turing-implementable strategies for solving CB problems.

Boolean functions are not, in general, invertible and finding the solution of a CB problem is not equivalent to finding the zeros of a vector function $f$ because the solution belongs to $\Zbb^n$ and not to $\Rbb^n$ or $\Cbb^n$. Nevertheless, we can still think of a strategy to find the solution of a CB problem that can be implemented in a specific machine. This means that the machine must be designed specifically to solve only such a CB problem. We will give explicit examples of such machines in Sec.~\ref{NP_sec} when we solve factorization and the NP-hard version of the subset-sum problem. We then define:
\begin{defn}
An \textit{Inverse Protocol} (IP) is that which finds a solution of a given CB problem by encoding the boolean system $f$ into a machine capable of accepting as input $\bvec$, and giving back as output $\yvec$, solution of $f(\yvec)=\bvec$. 
\end{defn}
Roughly speaking, the IP is a sort of ``inversion'' of $f$ using special-purpose machines. 

%%%%%%%%%%%%%%%%%%%%
\subsection{Algorithmic Boolean Problems}\label{Problem_subsection_algo}
%%%%%%%%%%%%%%%%%%%%

Of course, the definition of CB problems does not exaust all possible problems in Computer Science. We can also define the following class:
\begin{defn}
The \textit{Algorithmic Boolean} (AB) problems are those problems described by a collection of statements mathematically formalized by a control flow with boolean functions {\it and} appropriate control flow statements.
\end{defn} 
The CB problems are clearly a subclass of AB problems. However, it is not easy to say {\it a priori} if an AB problem can be reduced to a CB problem. For example, control flows including loops that terminate only if a condition is met cannot be always translated into a finite system $f$ of boolean functions. Moreover, the same problem that can be formulated either as CB or AB, may not have the same complexity in the respective formulation. For example, in order to reduce $f$ to a unique boolean system, if we consider a control flow $CF$ that requires the evaluation of $n_B$ boolean functions and $n_C$ conditional branches (i.e., the ``if-then'' statements that are a specific case of control flow statements), we may need resources increasing with the dimension of $f$ which, in turn, increases non-polynomially with $n_C$. More insights will be given in Sec. \ref{complexity_section}. 

We remark that by the term \textit{resources} we denote the amount of {\it space}, {\it time}, and {\it energy} employed by the machine to find the solution of a specific problem. Finally, we only acknowledge that there exist problems that cannot be classified as AB. These are beyond the scope of our present paper and therefore will not be treated here. 

%%%%%%%%%%%%%%%%%%%%
\section{Digital Memcomputing Machines}\label{DMM}
%%%%%%%%%%%%%%%%%%%%

We are now ready to define the digital versions of memcomputing machines (DMMs). These are a special class of Universal Memcomputing Machines (UMMs) we have introduced in Ref. \cite{UMM} in which the memprocessors (processors with memory) have only a finite number of possible states {\it after} the transition functions, $\delta$s, have been applied. Note that the transition functions can be implemented using, e.g., circuits (as we will show later) or some other physical mechanism. The important point is that they map $\Zbb^n$ into $\Zbb^m$, with $n$ and $m$ arbitrary integers.

%%%%%%%%%%%%%%%%%%%%
\subsection{Formal definition}\label{formalDMM}
%%%%%%%%%%%%%%%%%%%%

We recall that the UMM is an ideal machine formed by a bank of $m$ interconnected memory cells -- memprocessors -- with $m\leq\infty$. Its DMM subclass performs digital (logic) operations controlled by a control unit (see Fig.~\ref{DMM_figure_1}). The computation {\it with} and {\it in} memory can be sketched in the following way. When two or more memprocessors are connected, through a signal sent by the control unit, the memprocessors change their internal states according to both their initial states and the signal, thus giving rise to intrinsic parallelism (interacting memory cells simultaneously and collectively
change their states when performing computation) and functional polymorphism (depending on the applied signals, the same interacting
memory cells can calculate different functions) \cite{UMM}. The information overhead typical of UMMs will be defined at length in Sec.~\ref{Info_over_subsection}. 

\begin{figure}
	\centerline{
		\includegraphics[width=\columnwidth]{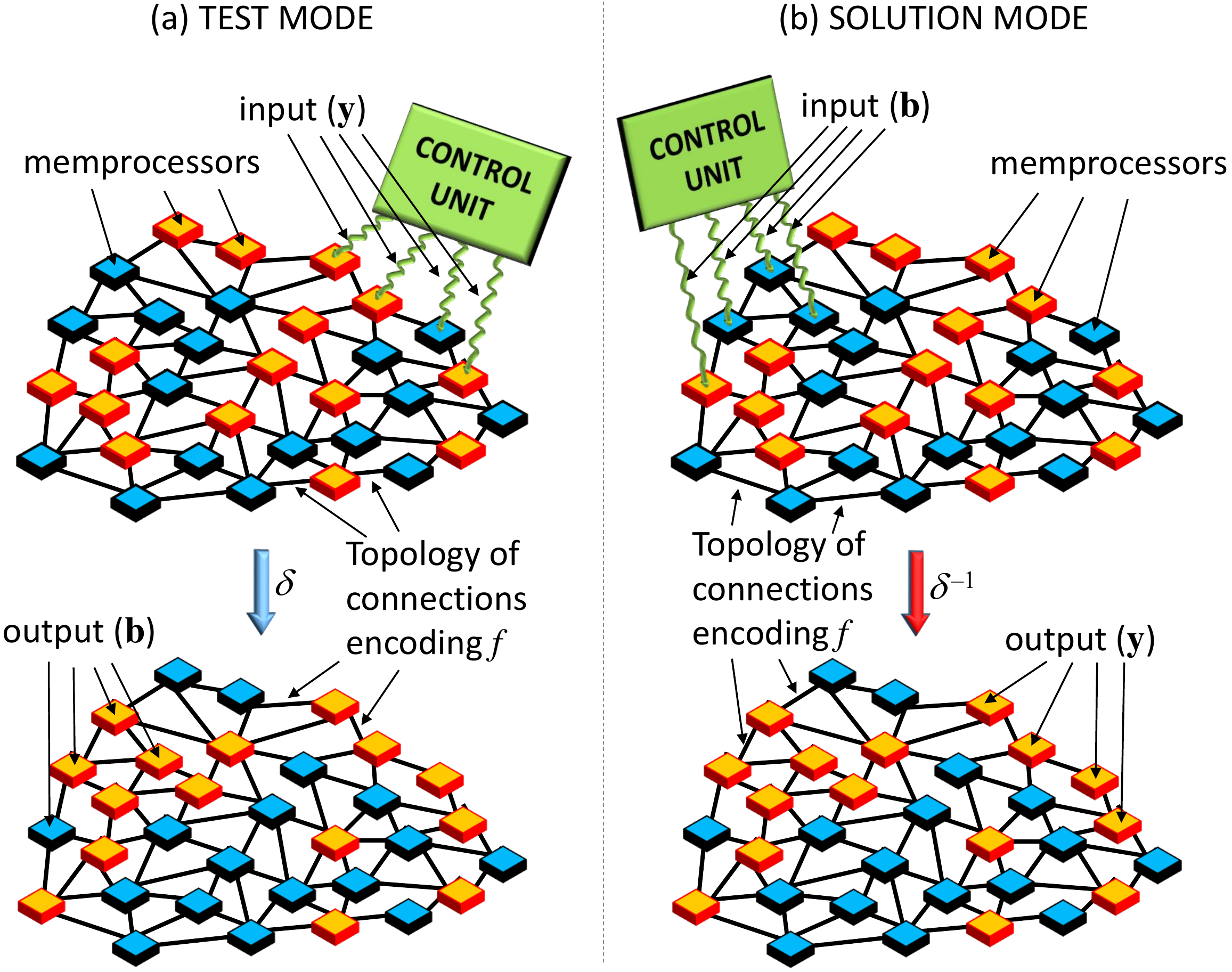}}
	\caption{\label{DMM_figure_1}Sketch of a DMM. The memprocessors are two-state interconnected elements that change their state according to both the external signal fed by the Control Unit and the signals from the other memprocessors through their connections. $\delta$ is the composition of all transition functions involved in the computation. The (a) panel shows the {\it test mode} for the verification of a given solution of a CB problem, while the (b) panel the {\it solution mode} for the IP implementation.}
\end{figure}

\begin{defn}
A \textit{DMM} is the eight-tuple
\begin{equation}
	DMM=(\Zbb,\Delta,{\cal P},S,\Sigma,p_0,s_0,F)\,.\label{UMMdef}
\end{equation}
Without loss of generality, we restrict the range to $\Zbb=\{0,1\}$ because the generalization to any finite number of states is trivial and does not add any major change to the theory. $\Delta$ is a set of functions
\begin{equation}
	\delta_\alpha:\Zbb^{m_\alpha}\backslash F\times {\cal P}\rightarrow \Zbb^{m\rq_\alpha}\times {\cal P}^2\times S\,,\label{functUMM}
\end{equation}
where $m_\alpha<\infty$ is the number of memprocessors used as input of (read by) the function $\delta_\alpha$, and $m\rq_\alpha<\infty$ is the number of memprocessors used as output (written by) the function $\delta_\alpha$; ${\cal P}$ is the set of the arrays of pointers $p_\alpha$ that select the memprocessors called by $\delta_\alpha$ and $S$ is the set of indexes $\alpha$; $\Sigma$ is the set of the initial states written by the input device on the computational memory; $p_0\in {\cal P}$ is the initial array of pointers; $s_0$ is the initial index $\alpha$,  and $F\subseteq \Zbb^{m_f}$ for some $m_f\in\mathbb{N}$ is the set of final states.
\end{defn}

%%%%%%%%%%%%%%%%%%%%
\subsection{Computational complexity classes}\label{complexity_section}
%%%%%%%%%%%%%%%%%%%%

\begin{figure}
	\centerline{
		\includegraphics[width=0.8\columnwidth]{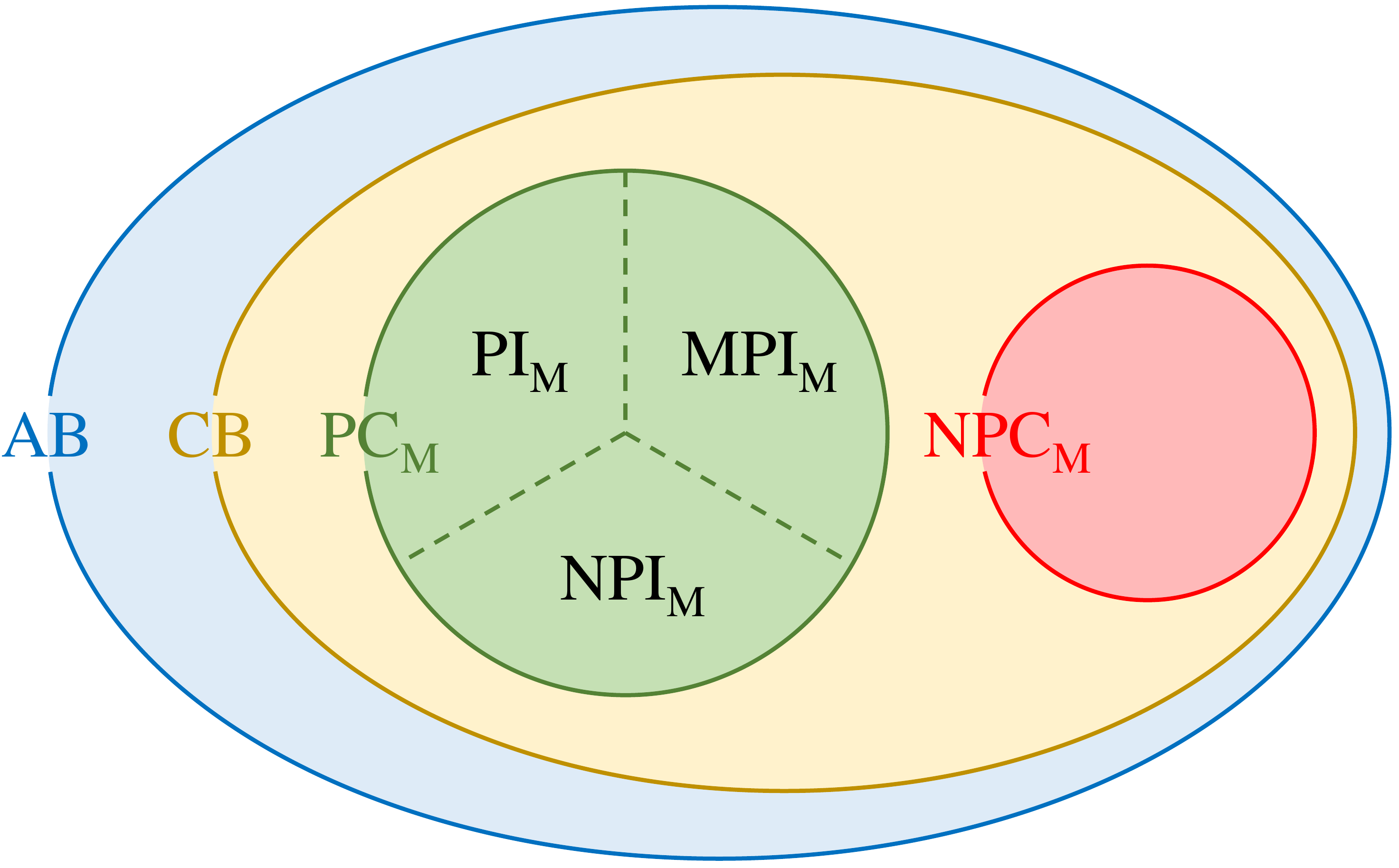}}
	\caption{\label{DMM_figure_2}Euler diagram of computatinal complexity classes with respect to digital memcomputing machines. Of all problems AB = Algorithmic Boolean, the subclass CB = Compact Boolean can be solved in time PC$_\text{M}$ = Polynomially Compact, NPC$_\text{M}$ = Non-Polynomially Compact, PI$_\text{M}$ = Polynomially Invertible, MPI$_\text{M}$ = Memcomputing Polynomially Invertible, NPI$_\text{M}$ = Non-Polynomially Invertible. Note that PC$_\text{M}$=PI$_\text{M}$ $\cup$ MPI$_\text{M}$ $\cup$ NPI$_\text{M}$.}
\end{figure}

We now use the notion of AB and CB problems introduced in Sec.~\ref{def_problem} to define the first two classes of computational complexity for a DMM. It is again worth recalling for the subsequent definitions that a DMM can be used also as a standard Turing machine when implemented in DP (see \cite{UMM} for more details). It implies that we can design DMMs to implement problems in their AB formulation as well as designing other DMMs implementing problems in their CB formulation.
\begin{defn}
A CB problem is said to be \textit{Polynomially Compact} for a DMM (PC$_\text{M}$) if, for a given $\yvec\in\Zbb^n$, {\it to verify} if $\yvec$ satisfies $f(\yvec)=\bvec$ both the DMM implementing $f(\yvec)$ and the DMM implementing the $CF$ from the AB formulation of the problem, require polynomial resources in $n$. 
\end{defn}
Conversely, 
\begin{defn}
a CB problem is said to be \textit{Non-Polynomially Compact} for a DMM (NPC$_\text{M}$) if, for a given $\yvec\in\Zbb^n$, {\it to verify} if $\yvec$ satisfies $f(\yvec)=\bvec$ the DMM implementing $f(\yvec)$ requires \textit{more than} polynomial resources in $n$ while the DMM implementing the $CF$ from the AB formulation of the problem requires only polynomial resources. 
\end{defn} 
We can further define other three computational complexity classes based on the complexity of the DP and IP.
\begin{defn}
A CB problem is said to be \textit{Polynomially Invertible} for a DMM (PI$_\text{M}$), if, {\it to find} a solution $\yvec$ of $f(\yvec)=\bvec$ with both the DP and the IP requires polynomial resources in $n=\dim(\bvec)$.
\end{defn}
On the other hand,
\begin{defn}
a CB problem is said to be \textit{Memcomputing Polynomially Invertible} for a DMM (MPI$_\text{M}$) if, {\it to find} a solution $\yvec$ of $f(\yvec)=\bvec$ with a DP requires {\it more than} polynomial resources in $n$, while the IP requires only polynomial resources.
\end{defn}
On the contrary,
\begin{defn}
a CB problem is said to be \textit{Non-Polynomially Invertible} for a DMM (NPI$_\text{M}$) if, {\it to find} a solution $\yvec$ of $f(\yvec)=\bvec$ with both the DP and the IP requires more than polynomial resources in $n=\dim(\bvec)$.
\end{defn}

Finally, we remark that the case in which the IP requires more than polynomial resources while the DP only polynomial 
resources belongs to NPC$_\text{M}$, and the case in which both require more than polynomial resources belongs to $\text{CB}\backslash\text{PC}_\text{M}$. However, both instances are not of interest to the present work. Therefore, we do 
not define them here. The Euler diagram of the complexity classes for DMMs is shown in Fig.~\ref{DMM_figure_2}.

%%%%%%%%%%%%%%%%%%%%
\subsection{Topological implementation of IP within DMMs}\label{implementation_section}
%%%%%%%%%%%%%%%%%%%%

A DMM can be designed either to implement the DP or the IP. We discuss only the implementation of the latter case since the implementation of the DP is the same as for Turing machines so it does not add anything new. 

We focus on a given CB problem characterized by the boolean system $f$. Since $f$ is composed of boolean functions, we can map $f$ into a boolean circuit composed of logic gates. Then, implementing $f$ into a DMM can be formally done by mapping the boolean circuit into the connections of the memprocessors. The DMM can then work in two different modes: test and solution modes (see Fig.~\ref{DMM_figure_1}). 

In the test mode (see panel (a) of Fig.~\ref{DMM_figure_1}) the control unit feeds the appropriate memprocessors with a signal encoding $\yvec$. In this way, the first transition function $\delta_\alpha$ receives its input. Performing the composition of all transition functions, $(\delta_\zeta\circ\cdots\circ\delta_\alpha)(\yvec)$, we obtain the output $f(\yvec)$ to be compared against $\bvec$ to determine whether or not $\yvec$ is a solution of the CB problem. In this case, the transition function $\delta=\delta_\zeta\circ\cdots\circ\delta_\alpha$ represents the encoding of $f$ through the {\it topology} of the connections of the memprocessors. 

Conversely, in the solution mode (see panel (b) of Fig.~\ref{DMM_figure_1}) the control unit feeds the appropriate memprocessors with a signal encoding $\bvec$. The first transition function $\delta^{-1}_\alpha$ receives its input, and the composition of all transitions functions, $(\delta^{-1}_\zeta\circ\cdots\circ\delta^{-1}_\alpha)(\bvec)$, produces the output $\yvec$. 

The transition function $\delta^{-1}=\delta^{-1}_\zeta\circ\cdots\circ\delta^{-1}_\alpha$ still represents the 
encoding of $f$ through the topology of the connections of the memprocessors. However, it works as some system $g$ such that $g(\bvec)=\yvec$. Note that $g$ is not strictly the inverse of $f$ because it may not exist as a boolean system. For this reason we call $g$ the {\it topological inverse} of $f$ and $\delta^{-1}$ the {\it topological inverse transition function} of $\delta$.

%%%%%%%%%%%%%%%%%%%%
\subsection{Some remarks}\label{Remarks_CB_section}
%%%%%%%%%%%%%%%%%%%%

We add here some remarks to better understand how DMMs operate. As we have seen, the DMM works by encoding a boolean system $f$ onto the {\it topology} of the connections. It means that to solve a given problem with a given input length we must use a given DMM formed by a number of memprocessors and topology of connections directly related to $f$ and the input length. Furthermore, we notice that the possibility of implementing the IP to solve a CB problem into a DMM is ultimately related to its intrinsic parallelism~\cite{UMM}. Finally, even if the realization of such particular DMM is not unique (there may be different types of memprocessors and topologies that solve such a problem) it is, however, a \textit{special purpose} machine designed to solve a specific problem. 

Conversely, the Turing machine is a general purpose machine in the sense that to solve a given problem we have to provide a set of instructions to be computed. In other words, the Turing machine can only solve problems employing the DP. It is also worth remarking that DMMs are more than standard neural networks (even if a neural network is a special case of DMMs). In fact, artificial neural networks do not fully map the problems into the connectivity of neurons. On the contrary, they are in some measure ``general purpose'' machines that need to be trained to solve a given problem. The DMM transfers the network training step into the choice of network topology. 

Finally, we provide some considerations on the possible strategy to device cryptographic schemes for DMMs. As we show in Sec.~\ref{NP_sec}, NP problems such as factorization or the subset-sum problem, can be efficiently solved with DMMs implementing the IP. It means that the standard cryptosystems like the RSA~\cite{RSA} or others based on the notion that 
NP problems provide possible one-way functions~\cite{Goldreich}, can be broken by DMMs. We then have to look at new cryptosystems beyond the Turing machine paradigm. As a guideline we notice here that a cryptosystem for DMMs should be based on problems belonging either to the class NPI$_\text{M}$ or NPC$_\text{M}$ that have both the DP and IP scaling more than polynomial in the length of the input. However, the class of the NPI$_\text{M}$ problems could be empty if all the PC$_\text{M}$ problems can be solved in polynomial time by the DMM employing the IP. In this case, we should instead focus on NPC$_\text{M}$ problems. 

%%%%%%%%%%%%%%%%%%%%
\subsection{Information Overhead}\label{Info_over_subsection}
%%%%%%%%%%%%%%%%%%%%

Using the definition given in the previous sections, we can now formally define a concept that we already have introduced in \cite{UMM}: the {\it information overhead}. We first stress that this information is \textit{not stored} in any memory unit. The stored information is the Shannon self-information that is equal for both Turing machines and UMMs. Instead, it represents the {\it extra} information that is \textit{embedded} into the topology of the connections of the DMMs. Thus we have:

\begin{defn}\label{IO_def}
	The \textit{information overhead} is the ratio 
	\begin{equation}\label{IO}
	    I_O =  \frac{\sum_i ( m^U_{\alpha_i} + m\rq^U_{\alpha_i} )}{\sum_j ( m^T_{\alpha_j} + m\rq^T_{\alpha_j} )},
	\end{equation}
	where $m^T_{\alpha_j}$ and $m\rq^T_{\alpha_j}$ are the number of memprocessors read and written by the transition function $\delta^T_{\alpha_j}$ that is the transition function of a DMM formed by interconnected memprocessors with {\it topology} of connections related to the CB problem, while $m^U_{\alpha_i}$ and $m\rq^U_{\alpha_i}$ are the number of memprocessors read and written by the transition function $\delta^U_{\alpha_i}$ that is the transition function of a DMM formed by the {\it union} of non-connected memprocessors. The sums run over all the transition functions used to find the solution of the problem.
\end{defn} 

From this definition, if we use the IP, $\delta^T_{\alpha_j} = \delta^{-1}_{\alpha_j}$ defined in section \ref{implementation_section}. Conversely, since $\delta^U_{\alpha_i}$ is the transition function of non-connected memprocessors the unique way to find a solution is by employing the DP. Therefore, problems belonging to the class MPI$_\text{M}$ have information overhead that is more than polynomial, i.e., the topology encodes or ``compresses'' much more information about the problem than the DP. It is worth noticing again that this information is something related to the problem, i.e., related to the structure of $f$ (topology of the memprocessor network), and not some actual data stored by the machine. Finally, due to its different nature, the information overhead that we have defined does not satisfy the standard mathematical properties of information measures (see Sec.~\ref{Info_acc_subsection}). Instead, it provides a practical measure of the importance of the topology in a given computing network.

%%%%%%%%%%%%%%%%%%%%
\section{Dynamical Systems Picture}\label{Dyna_section}
%%%%%%%%%%%%%%%%%%%%

In order to provide a practical, physical route to implement DMMs, we now reformulate them in terms of dynamical systems. This formulation will also be useful for a formal definition of accessible information and information overhead (the latter already introduced in \cite{UMM}) and at the same time to clearly point out their scalability and main differences with the most common definitions of parallel Turing machines. Note that a dynamical systems formulation has also been used to discuss 
the computational power of analog neural networks~\cite{Havabook,HavaFish}, which can be classified as a type of memcomputing machines that have no relation between topology and the problem to solve.  Here, we extend it to DMMs proper. 

%%%%%%%%%%%%%%%%%%%%
\subsection{Dynamical systems formulation of DMMs}\label{Dyna_DMM}
%%%%%%%%%%%%%%%%%%%%

We consider the transition function $\delta_\alpha$ defined in \eqref{functUMM} and $p_\alpha,p\rq_\alpha\in {\cal P}$ being the arrays of pointers $p_\alpha=\{i_1,...,i_{m_\alpha} \}$ and $p\rq_\alpha=\{j_1,...,j_{m\rq_\alpha} \}$. Furthermore, we describe the state of the DMM (i.e., the state of the network of memprocessors) using the vector $\xvec\in\Zbb^n$. Therefore, $\xvec_{p_\alpha}\in \Zbb^{m_\alpha}$ is the vector of states of the memprocessors selected by $p_\alpha$. Then, the transition function $\delta_\alpha$ acts as:
\begin{equation}
	\delta_\alpha(\xvec_{p_\alpha})=\xvec\rq_{p\rq_\alpha}\,.
	\label{delta}
\end{equation}
In other words, $\delta_\alpha$ reads the states $\xvec_{p_\alpha}$ and writes the new states $\xvec\rq_{p\rq_\alpha}$, both belonging to $\Zbb^n$. Equation \eqref{delta} points out that the transition function $\delta_\alpha$ {\it simultaneously} acts on a set of memprocessors (intrinsic parallelism). 

The intrinsic parallelism of these machines is the feature that is at the core of their power. It is then very important to understand its mechanism and what derives from it. To 
analyze it more in depth, we define the time interval $\I_\alpha=[t,t+\T_\alpha]$ that $\delta_\alpha$ takes to perform the transition. We can describe mathematically the dynamics during $\I_\alpha$  by making use of the dynamical systems framework \cite{perko_01}. At the instant $t$ the control unit sends a signal to the computational memory whose state is described by the vector $\xvec(t)$. The dynamics of the DMM within the interval $\I_\alpha$ between two signals sent by the control unit is described by the flow $\phi$ (the definition of flow follows from the dynamical systems formalism \cite{perko_01})  
\begin{equation}
\xvec(t\rq{}\in \I_\alpha)=\phi_{t\rq-t}(\xvec(t)) \label{UMMdyn},
\end{equation} 
namely {\it all} memprocessors interact at each instant of time in the interval $\I_\alpha$ such that 
\begin{align}
&\xvec_{p_\alpha}(t)=\xvec_{p_\alpha}\\
&\xvec_{p\rq_\alpha}(t+\T_\alpha)=\xvec\rq_{p\rq_\alpha}
\end{align}

%%%%%%%%%%%%%%%%%%%%
\subsection{Parallel Turing Machines}\label{PTM_subsection}
%%%%%%%%%%%%%%%%%%%%

In \cite{NANO_15} we have briefly discussed the parallelism in today\rq{}s computing architectures, i.e., Parallel Turing Machines (PTM). Here, for the sake of completeness, we summarize the main results useful to compare against the {\it intrinsic} parallelism of DMMs.
  
Parallelism in standard computing machines can be viewed from two perspectives: the practical and the theoretical one. The theoretical approach is still not well assessed and there are different attempts to give a formal definition of PTMs \cite{Kozen_76,Fortune_78,Wolfram_84,wiedermann_84,Karp_88,Worsch_97,Worsch_99,Worsch_12}. Irrespective, the PTM often results in an ideal machine that cannot be practically built \cite{wiedermann_84,Worsch_97}. Therefore, we are more interested in a description of the practical approach to parallelism. 

The definition we give here includes some classes of PTMs, in particular the Cellular Automata and non-exponentially growing Parallel Random Access Machines \cite{Kozen_76,Fortune_78,Wolfram_84,Worsch_12}. We consider a fixed number of (or at most a polynomially increasing number of) central processing units (CPUs) that perform some tasks in parallel. In this case each CPU can work with its own memory cash or accesses a shared memory depending on the architecture. In any case, in {\it practical} Parallel Machines (PM), all CPUs are synchronized, each of them performs a task in a time $\T_\text{PM}$ (the synchronized clock of the system of CPUs), and at the end of the clock cycle all CPUs share their results and follow with the subsequent task. We can describe mathematically also this picture within the dynamical system framework \cite{perko_01} to point out the difference with the intrinsic parallelism of memcomputing. 

Let us consider the vector functions $\svec(t)=[s_1(t),...,s_{n_s}(t)]$ and $\kvec(t)=[k_1(t),...,k_{n_k}(t)]$ defining respectively the states of the $n_s$ CPUs and the symbols written on the total memory formed by $n_k$ memory units. As long as the CPUs perform their computation, at each clock cycle they act {\it independently} so there are $n_s$ independent flows describing the dynamics during the computation of the form $(s_j(t+\T_\text{PM}),k_{j_w}(t+\T_\text{PM}))=\phi^j_{\T_\text{PM}}(s_j(t),\kvec(t))$ where $k_{j_w}$ is the memory unit written by the $j$-th CPU. Since the $j$-th CPU just reads the memory $\kvec(t)$ at only the time $t$ and not during the interval $\I_\text{PM}=]t,t+\T_\text{PM}]$, and it does not perform any change on it apart from the unit $k_{j_w}$, the evolution of the entire state during $\I_\text{PM}$ is completely determined by the set of independent equations
\begin{equation}
	(s_j(t\rq{}\in \I_\text{PM}),m_{j_w}(t\rq{}\in \I_\text{PM}))=\phi^j_{t\rq-t}(s_j(t),\kvec(t))\label{PTMdyn}.
\end{equation} 
A quick comparison with Eq.~\ref{UMMdyn} shows the fundamental difference with memcomputing machines: in each interval $\I_\text{PM}$ the $n_s$ CPUs do not interact in any way and their dynamics are {\it independent}. 

%%%%%%%%%%%%%%%%%%%%
\subsection{Accessible Information}\label{Info_acc_subsection}
%%%%%%%%%%%%%%%%%%%%

Let us now consider a DMM and a PTM having the same number, $m$, of memprocessors, and standard processors, respectively, and taking the same time $\T$ to perform a step of computation. Moreover, we can also assume that at time $t$ and $t+\T$, the computing machines (whether DMM or PTM) have the state of the memory units that can be either 0 or 1. Therefore, the number of all possible initial or final configurations of both the DMM and PTM is $2^{m}$, and the Shannon self-information is equal for both, namely $I_S=-\log_2[2^{-m}]=m$. We introduce the 

\begin{defn}
	The \textit{accessible information} $I_A$ is the volume of the configuration space explored by the machine during the computation, i.e., during the interval $\I_\T$.
\end{defn}

We remark that the concept and definition of accessible information (and also of information overhead given in section \ref{Info_over_subsection}), even if it follows from principles of statistical mechanics, it is not the standard information measure in the sense that it does not satisfy the usual properties of the information measures. On the other hand, our definition is able to highlight relevant features useful to point out differences between DMMs (or UMMs in general) and Turing machines. For instance, from \eqref{PTMdyn}, the PTM is equivalent to a system of $m$ non-interacting objects, so from the principles of statistical mechanics \cite{landau2013statistical} the number of total configurations is the sum of the configurations allowed for each individual element, i.e., $2m$, and hence the volume of the configuration space explored during the computation is also $2m$. The accessible information in this case is then $I_A\propto2m$.  On the other hand, in the case of DMMs, we have that the memprocessors interact and follow the dynamics \eqref{UMMdyn}, therefore the volume of configuration space explored by a DMM during $\I$ is $2^{m}$ and $I_A\propto2^{m}$.

%%%%%%%%%%%%%%%%%%%%
\subsection{Information Overhead vs. Accessible Information}\label{Info_OA_subsection}
%%%%%%%%%%%%%%%%%%%%

We have introduced the concept of information overhead in \cite{UMM} and formalized it in the problem/solution picture in Sec.~\ref{Info_over_subsection}. In this section we discuss its relationship with the accessible information making use of the dynamical systems picture. 

As described above, the DMMs can access, at each computational step, a configuration space that may grow exponentially. Therefore, a DMM, even if it starts from a well defined configuration and ends in another well defined configuration both belonging to $\Zbb^m$, during the computation its \textit{collective state} (i.e., the state of the memprocessor network) results in a superposition of all configurations allowed by the DMM. 

We consider a unit proportionality coefficient between the volume of the configuration space explored by the machine and the accessible information as defined in Sec.~\ref{Info_acc_subsection}. Now, the accessible information is always larger or equal to the number of memprocessors involved in the computational step. In fact, we have $2m\le I_A \le 2^{m}$ and $(m_\alpha+m\rq_\alpha)\le 2m$ and then $(m_\alpha+m\rq_\alpha)\le I_A$. Therefore, using equation \eqref{IO}, we have  

\begin{equation}\label{IO_1}
	I_O \le \frac{\sum_i I^U_{Ai}}{\sum_j ( m^T_{\alpha_j} + m\rq^T_{\alpha_j} )},
\end{equation}
where $I^U_{Ai}$ is the accessible information of the computational step $i$ of a DMM formed by a union of non-connected memprocessors. 

We can further notice that $I^U_{Ai} = 2m^U_i \ge (m^U_{\alpha_i} + m\rq^U_{\alpha_i}) $ where $m^U_i$ is the total number of memprocessors involved in the computational step $i$. Conversely, the accessible information of the computational step $j$ of a DMM formed by interconnected memprocessors is $I^T_{Aj} = 2^{m^T_j}$. 

Now, we consider a CB problem belonging to PC$_\text{M}$ with $n=\dim(\bvec)$ and $f$ the associated boolean system. By the definitions in Sec.~\ref{complexity_section}, the number of memprocessors involved in the computation (either in test or solution mode) is a polynomial function of $n$. Moreover, the number of steps to solve a problem in MPI$_\text{M}$ is polynomial for the case of IP (namely when interconnected memprocessors encode $f$ into the topology of the connections) while it could be more than polynomial in the case of DP (for examples, the union of non-connected memprocessors). Then, there exist two positive polynomial functions $P(n)$ and $Q(n)$ such that $P(n)\sum_i I^U_{Ai}= Q(n)\sum_j I^T_{Aj}$. We can substitute this relation into Eq.~\eqref{IO_1} and have (since all quantities are related to the same machine we can suppress the superscript $T$): 
\begin{equation}\label{IO_2}
I_O \le \frac{Q\sum_j I_{Aj}}{P\sum_j ( m_{\alpha_j} + m\rq_{\alpha_j} )}.
\end{equation}

It is worth noticing that the relation \eqref{IO_2} is valid for any type of topology, including the one with no connections. In fact, in the case of no connections and computational steps that involve all memprocessors (which is equivalent to a PTM) we have that $P=Q$, $I_{Aj} = 2m_j = (m_{\alpha_j} + m\rq_{\alpha_j}) $ and $I_O \le 1$, namely, no information has been encoded (compressed) into the topology of the machine. 

The other limit case is that of connections not related to the problem (like in neural networks). In this case, at least one of $P$ or $Q$ cannot be polynomial but the ratio $P/Q$ must be more than polynomial in $n$ and the $I_O$ is maximized by a polynomial function of $n$.
%%%%%%%%%%%%%%%%%%%%
\subsection{Equilibrium points and IP }\label{equilibrium_subsection}
%%%%%%%%%%%%%%%%%%%%

The information overhead and accessible information defined in the previous section can be interpreted as operational measures of the computing power of a DMM. However, another important feature to characterize the computing power of these machines is the composition of its phase space, i.e., the $n$-dimensional space in which $\xvec(t)$ is a trajectory (not to be confused with the configuration space). In fact, the extra information embedded into a DMM by encoding $f$ onto the connection topology strongly affects the dynamics of the DMM, and in particular its global phase portrait, namely the phase space and its regions like attraction basins, separatrix, etc. \cite{perko_01}. 

To link all these concepts, we first remark that the dynamical system describing a DMM should have the following 
properties:
\begin{itemize}
	\item Each component $x_j(t)$ of $\xvec(t)$ has initial conditions $x_j(0)\in X$, where $X$ is the phase space and also a metric space.
	\item To each configuration belonging to $\Zbb^m$, one or more equilibrium points $\xvec_s\in X$ can be associated, and the system converges exponentially fast to these equilibria~\cite{hale_2010_asymptotic}. 
	\item The stable equilibria $\xvec_s\in X$ are associated to the solutions of the CB problem.
	\item To consistently define the DMM as a \textit{digital} machine, the input and output of the DMM (namely $\yvec$ and $\bvec$ in test mode or $\bvec$ and $\yvec$ in solution mode) must be mapped into a set of parameters $\pvec$ (input) and equilibria $\xvec_s$ (output) such that $\exists \, \hat p,\hat x\in \Rbb$ and $|p_j-\hat p|=c_p$
    and $|x_{s_j}-\hat x|=c_x$ for some $c_p,c_x>0$ independent of $n_y=\dim(\yvec)$ and $n=\dim(\bvec)$. Moreover, if we indicate a polynomial function of $n$ of maximum degree $\gamma$ with $\poly_\gamma(n)$, then $\dim(\pvec)=\poly_{\gamma_p}(n)$ and $\dim(\xvec_s)=\poly_{\gamma_x}(n_b)$ in test mode or $\dim(\pvec)=\poly_{\gamma_p}(n_b)$ and $\dim(\xvec_s)=\poly_{\gamma_x}(n)$ in solution mode, with $\gamma_x$ and $\gamma_p$ independent of $n_b$ and $n$.
    \item Other stable equilibria, periodic orbits or strange attractors that we generally indicate with $\xvec_w(t) \in X$, not associated with the solution(s) of the problem, may exist, but their presence is either irrelevant or can be accounted for 
    with appropriate initial conditions.
    \item The system has a {\it compact global asymptotically stable attractor}~\cite{hale_2010_asymptotic}, which means that it exists a compact $J\subset X$ that attracts the whole space $X$ (for the formal definition see Sec.~\ref{global_attractor_section}).
	\item The system converges to equilibrium points exponentially fast starting from a region of the phase space whose measure is not zero, and which can decrease at most polynomially with the size of the system. Moreover, the convergence time can increase at most polynomially with the input size.
\end{itemize}
It is worth noticing that the last requirements are satisfied if the phase space is completely clustered in regions that are the attraction basins of the equilibrium points and, possibly, periodic orbits and strange attractors, if they exist. We also point out that a class of systems that has a global attractor is the one of \textit{dissipative dynamical systems} \cite{hale_2010_asymptotic}. We will give a more detailed description of them in Sec.~\ref{global_attractor_section} but we anticipate here that by ``dissipative'' we do not necessarily mean ``passive'': active systems 
can be dissipative as well in the functional analysis sense\cite{hale_2010_asymptotic}. We will also show in Sec.~\ref{equilibria_subsection} that all the examples we provide satisfy these properties with the added benefit that the only equilibrium points are the solutions of the problem.

We further define $V=\text{vol}(X)$ the hyper-volume of $V$ defined by some Lebesgue measure on $X$, $J_s\subseteq J$ the compact subset of the global attractor $J$ containing all equilibria $\xvec_s$ only and $J_w\subseteq J$ the compact subset of the global attractor $J$ containing all $\xvec_w(t)$ only. Therefore, we have $J=J_s\cup J_w$ and $J_s\cap J_w=\emptyset$. We also define the sub-hyper volumes $V_s = \text{vol}(X_s)$ where $X_s\subseteq X$ is the subset of $X$ attracted by $J_s$. Similarly, we have  sub-hyper volumes $V_w = \text{vol}(X_w)$ where $X_w\subseteq X$ is the subset of $X$ attracted by $J_w$. By definition of attractors, $J_s$ and $J_w$, we have that $\text{vol}(X_s\cap X_w)=0$ and since $J$ is a global attractor
\begin{equation}
V= V_s +V_w\,.
\end{equation}

Using these quantities we can define both the probabilities that the DMM finds a solution of a CB using the IP, or that it fails to find it as  
\begin{equation}
P_s = \frac{V_s}{V}, \hspace{1cm} 
P_w = \frac{V_w}{V}\,,
\end{equation}   
respectively.

Clearly, this analysis refers to a DMM that finds a solution $\yvec$ of a CB problem only in one step. In this case we say that the DMM works properly iff $P_s>\poly_\gamma^{-1}(n)$ for some $\gamma\in \Nbb$. 

On the other hand, we can have IPs that require several computational steps to find the solution. In this case, after each step of computation the control unit can add some extra input depending on the $\xvec_w$. Therefore, we can define the sequence of probabilities of succeeding or failing as 
\begin{align}
P^1_s = \frac{V^1_s}{V},& \hspace{1cm} P^1_w = \frac{V^1_w}{V}, \nonumber\\
&\hspace{5pt}...\\ 
P^h_s = \frac{V^h_s}{V},& \hspace{1cm} P^h_w = \frac{V^h_w}{V}. \nonumber
\end{align}
If the number of steps, $h$, required to find the solution is a polynomial function of $n$ such that $P^h_z > \poly_\gamma^{-1}(n)$, we then say that the DMM works properly, or more precisely it has found the solution of the given CB problem in MPI$_\text{M}$ with polynomial resources in $n$ (cf. Fig.~\ref{DMM_figure_2}).
%\begin{figure}
%	%\centerline{
%	%\includegraphics[width=\columnwidth]{DMM_figure_2.pdf}}
%	\caption{\label{DMM_figure_3}.....}
%\end{figure}  

%%%%%%%%%%%%%%%%%%%%
\section{Self-organizing logic gates}\label{SOLG_section}
%%%%%%%%%%%%%%%%%%%%

In the previous sections we have given the main definitions and properties characterizing DMMs as a mathematical entity. Within the theory of dynamical systems we have related them to possible physical systems that perform computation using the IP. From all this we then conclude that there is  no mathematical limitation in presupposing the existence of a system with the properties of a DMM.  
With the aim of finding an actual physical system that satisfies all requirements of a DMM, we give here the description of a new type of logic gates that assembled in a circuit provide a possible realization of DMMs. These logic gates can be viewed both as a mathematical and as a physical realization of a DMM. 

To summarize the contents of this section, we first provide an abstract picture of the Self-Organizable Logic Gates (SOLGs) in Sec.~\ref{SOLG_subsection}. In Sec.~\ref{SOLC_subsection} we discuss their assembling into Self-Organizable Logic Circuits (SOLCs). From these abstract concepts, in Sec.~\ref{MEM_SOLG_subsection} we describe a possible realization of SOLGs using electronic components, and in Sec.~\ref{MEM_SOLC_subsection} we discuss the auxiliary circuitry necessary to obtain only stable solutions. In Sec.~\ref{Math_SOLC_section}, starting from the SOLCs equations, we will perform a detailed mathematical analysis in order to prove that these circuits indeed satisfy the requirements to be classified as DMMs. 

%%%%%%%%%%%%%%%%%%%%
\subsection{General Concept}\label{SOLG_subsection}
%%%%%%%%%%%%%%%%%%%%

\begin{figure}
	\centerline{
		\includegraphics[width=.7\columnwidth]{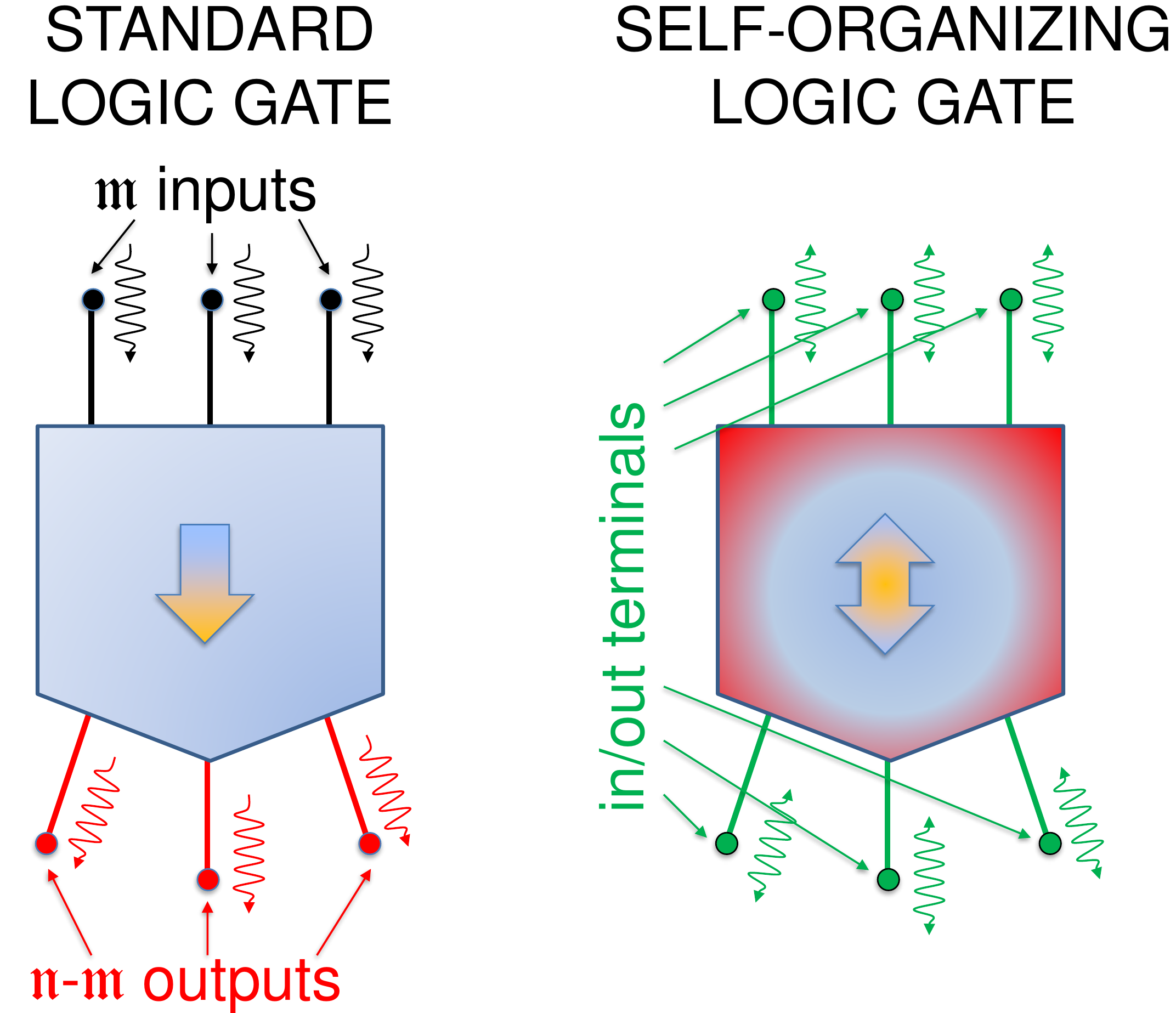}}
	\caption{\label{DMM_figure_3}Left panel: sketch of a standard logic gate. Right panel: sketch of a self-organizing logic gate.}
\end{figure}  

We briefly recall that a standard $\nter$-terminal logic gate with $\mter$ inputs and $\nter-\mter$ outputs follows the scheme in the left panel of Fig.~\ref{DMM_figure_3}. The input  terminals receive the signals, the gate processes the inputs, and finally sends to the output terminals the result of the computation. This is a {\it sequential} DP. 

Instead, we can devise logic gates, such as AND, OR, XOR, NOT, NOR, NAND and any other one in such a way that they work as standard gates by varying the input and obtaining the output, while ``in reverse'', by varying the output, they dynamically provide an input {\it consistent} with that output. We call the objects having these properties \textit{self-organizing logic gates} (SOLGs). 

SOLGs can use {\it any} terminal {\it simultaneously} as input or output, i.e., signals can go in and out at the same time at any terminal resulting in a superposition of input and output signals as depicted in the right panel of Fig.~\ref{DMM_figure_3}. The gate changes {\it dynamically} the outgoing components of the signals depending on the incoming components according to some rules aimed at satisfying the logic relations of the gate.

\begin{figure}
	\centerline{
		\includegraphics[width=\columnwidth]{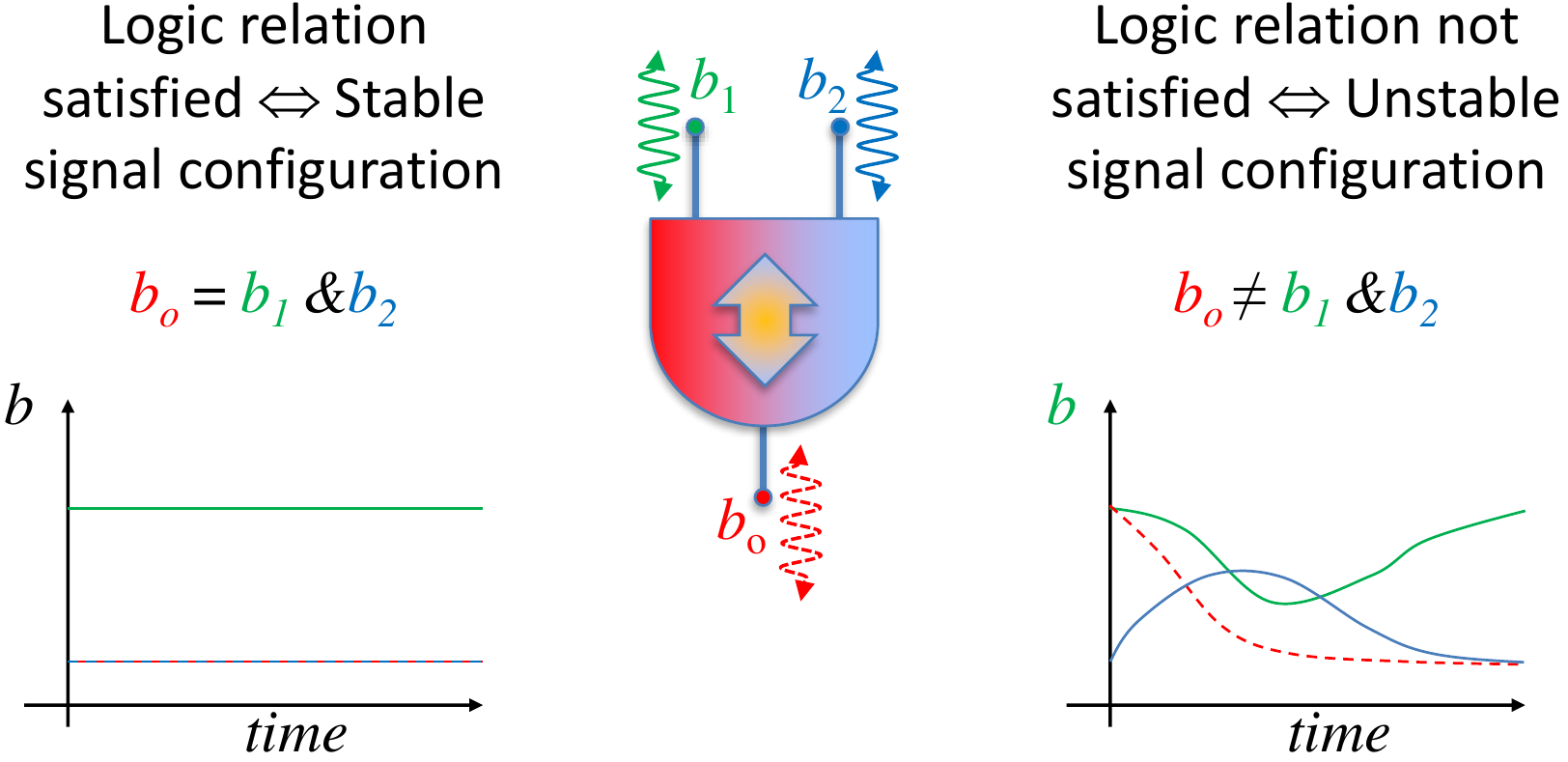}}
	\caption{\label{DMM_figure_4}Left panel: stable configuration of an SO-AND. Right panel: unstable configuration of an SO-AND.}
\end{figure}  

As depicted in Fig.~\ref{DMM_figure_4} a SOLG (in the actual example of Fig.~\ref{DMM_figure_4} the self-organizing (SO) AND gate is reported) can have either stable configurations (left panel) or unstable (right panel) configurations. In the former case, the configuration of the signals at the terminals satisfies the logic relation required (in the case of  Fig.~\ref{DMM_figure_4} the logic AND relation) and the signals would then remain constant in time. Conversely, if the signals at a given time do not satisfy the logic relation, we have the unstable configuration: the SOLG drives the outgoing components of the signal to finally obtain a stable configuration. 

%%%%%%%%%%%%%%%%%%%%
\subsection{Self-Organizable Logic Circuits}\label{SOLC_subsection}
%%%%%%%%%%%%%%%%%%%%

\begin{figure}
	\centerline{
		\includegraphics[width=\columnwidth]{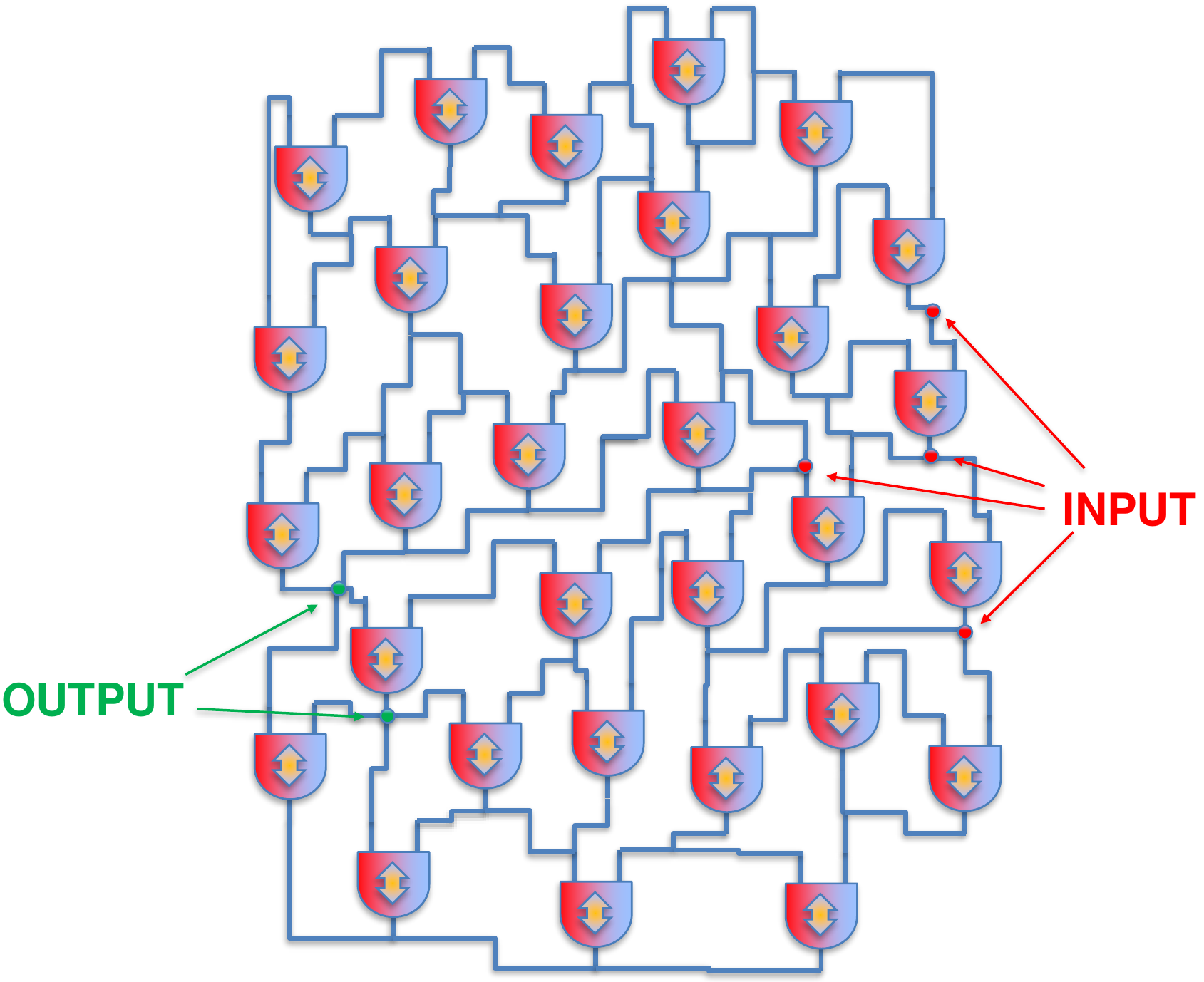}}
	\caption{\label{DMM_figure_5}Self-organizing logic circuit formed by a network of SO-AND gates. The external inputs are sent to some nodes related to the computational task at hand. The self-organizing circuit organizes itself by finding a stable configuration that satisfies the logic proposition and then the solution is read at the output nodes.}
\end{figure}  

Once we have introduced the SOLGs it is natural to put them together to obtain a DMM that can work either in test mode or in solution mode (see Sec.~\ref{implementation_section}). We then introduce the concept of \textit{self-organizing logic circuit} (SOLC) as a circuit composed of SOLGs connected together with the appropriate topology (see for example Fig.~\ref{DMM_figure_5} for a SOLC composed of only SO-AND gates). At each node of a SOLG an external input signal can be provided and the output can be read at other nodes of the SOLC. The connections of the circuit are related to a specific computational task required by the circuit. The topology of the connections and the specific logic gates used are not necessarily unique and can be derived from standard Boolean logic circuit theory, meaning that for a given CB problem, mapping $f$ into the connections can be done by using standard boolean relations. In this way a SOLC represents a possible realization of a DMM that can work in either test or solution mode.

%%%%%%%%%%%%%%%%%%%%
\subsection{Examples of Electronic-Based SOLGs}\label{MEM_SOLG_subsection}
%%%%%%%%%%%%%%%%%%%%

\begin{figure}
	\centerline{
		\includegraphics[width=.8\columnwidth]{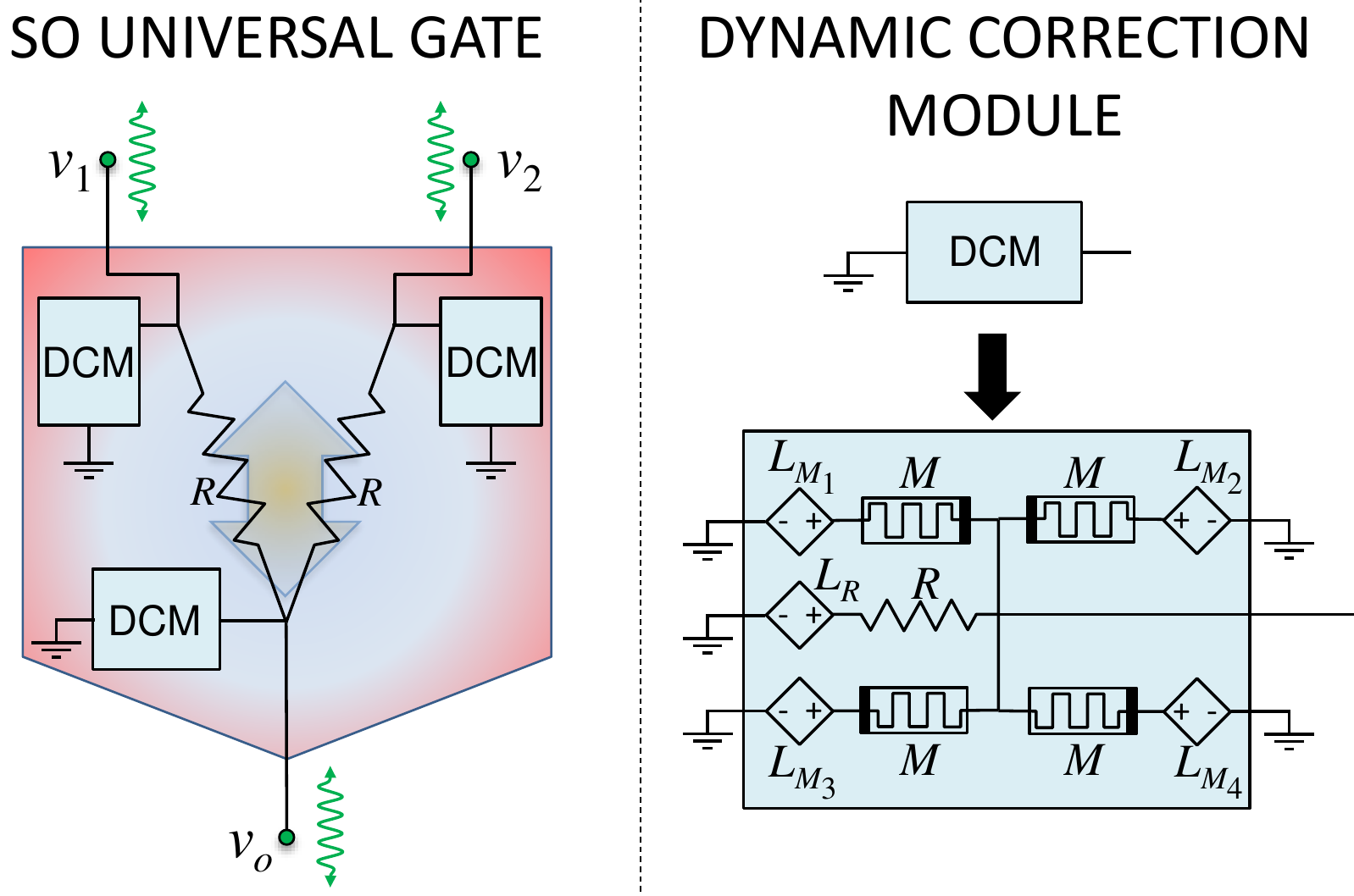}}
	\caption{\label{DMM_figure_6}Self-organizing universal gate (left panel) formed by dynamic correction modules (right panel). The memristors $M$ have minimum and maximum resistances $R_{on}$ and $R_{off}$, respectively. $R=R_{off}$ and $L_\gamma$ are linear functions driving the voltage-controlled voltage generators and depending only on $v_1$, $v_2$ and $v_o$ as reported in table \ref{table1} for different gates.}
\end{figure}

\begin{table*}
	\caption{Universal SO Gate Parameters}
	\label{table1}
	\begin{center}
\begin{tabular}{l r|r|r|r||r|r|r|r||r|r|r|r|}
%	\cline{4-11} & \multicolumn{2}{ c  |}{} & \multicolumn{8}{ c  }{\multirow{2}{*}{SO AND}} & \multicolumn{2}{| c  }{} \\
%	\multicolumn{3}{ c  |}{} & \multicolumn{8}{ c  |}{} \\
	\multicolumn{13}{ c  }{\multirow{2}{*}{SO AND}} \\ \\
	\cline{2-13} & \multicolumn{4}{ |c|| }{Terminal 1} & \multicolumn{4}{ c|| }{Terminal 2} & \multicolumn{4}{ c| }{Out Terminal} \\
	\cline{2-13} & \multicolumn{1}{|r|}{$a_1$}  & $a_2$ & $a_o$ & $dc$    & $a_1$  & $a_2$ & $a_o$ & $dc$    & $a_1$  & $a_2$ & $a_o$ & $dc$ \\ \hline 
	\multicolumn{1}{|l|}{$L_{M_1}$}  &
	\multicolumn{1}{|r|}{0}  & -1 & 1 & $v_c$    & -1 & 0 & 1 & $v_c$    & 1 & 0 & 0 & 0 \\ \hline 
	\multicolumn{1}{|l|}{$L_{M_2}$}  &
	\multicolumn{1}{|r|}{1}  & 0 & 0 & 0     & 0 & 1 & 0 & 0     & 0 & 1 & 0 & 0 \\ \hline 
	\multicolumn{1}{|l|}{$L_{M_3}$}  &
	\multicolumn{1}{|r|}{0}  & 0 & 1 & 0     & 0 & 0 & 1 & 0     & 0 & 0 & 1 & 0 \\ \hline 
	\multicolumn{1}{|l|}{$L_{M_4}$}  &
	\multicolumn{1}{|r|}{1}  & 0 & 0 & 0     & 0 & 1 & 0 & 0     & 2 & 2 & -1 & -2$v_c$ \\ \hline 
	\multicolumn{1}{|l|}{$L_{R}$}  &
	\multicolumn{1}{|r|}{4}  & 1 & -3 & -$v_c$   & 1 & 4 & -3 & -$v_c$   & -4 & -4 & 7 & 2$v_c$ \\ \hline 
	\\
	\multicolumn{13}{ c  }{\multirow{2}{*}{SO OR}} \\ \\
	\cline{2-13} & \multicolumn{4}{ |c|| }{Terminal 1} & \multicolumn{4}{ c|| }{Terminal 2} & \multicolumn{4}{ c| }{Out Terminal} \\
	\cline{2-13} & \multicolumn{1}{|r|}{$a_1$}  & $a_2$ & $a_o$ & $dc$    & $a_1$  & $a_2$ & $a_o$ & $dc$    & $a_1$  & $a_2$ & $a_o$ & $dc$ \\ \hline 
	\multicolumn{1}{|l|}{$L_{M_1}$}  &
	\multicolumn{1}{|r|}{0}  & 0 & 1 & 0     & 0 & 0 & 1 & 0     & 0 & 0 & 1 & 0 \\ \hline 
	\multicolumn{1}{|l|}{$L_{M_2}$}  &
	\multicolumn{1}{|r|}{1}  & 0 & 0 & 0     & 0 & 1 & 0 & 0     & 2 & 2 & -1 & 2$v_c$ \\ \hline 
	\multicolumn{1}{|l|}{$L_{M_3}$}  &
	\multicolumn{1}{|r|}{0}  & -1 & 1 & -$v_c$    & -1 & 0 & 1 & -$v_c$    & 1 & 0 & 0 & 0 \\ \hline 
	\multicolumn{1}{|l|}{$L_{M_4}$}  &
	\multicolumn{1}{|r|}{1}  & 0 & 0 & 0     & 0 & 1 & 0 & 0     & 0 & 1 & 0 & 0 \\ \hline 
	\multicolumn{1}{|l|}{$L_{R}$}  &
	\multicolumn{1}{|r|}{4}  & 1 & -3 & $v_c$   & 1 & 4 & -3 & $v_c$   & -4 & -4 & 7 & -2$v_c$ \\ \hline 
	\\
	\multicolumn{13}{ c  }{\multirow{2}{*}{SO XOR}} \\ \\
	\cline{2-13} & \multicolumn{4}{ |c|| }{Terminal 1} & \multicolumn{4}{ c|| }{Terminal 2} & \multicolumn{4}{ c| }{Out Terminal} \\
	\cline{2-13} & \multicolumn{1}{|r|}{$a_1$}  & $a_2$ & $a_o$ & $dc$    & $a_1$  & $a_2$ & $a_o$ & $dc$    & $a_1$  & $a_2$ & $a_o$ & $dc$ \\ \hline 
	\multicolumn{1}{|l|}{$L_{M_1}$}  &
	\multicolumn{1}{|r|}{0}  & -1 & -1 & $v_c$   & -1 & 0 & -1 & $v_c$    & -1 & -1 & 0 & $v_c$ \\ \hline 
	\multicolumn{1}{|l|}{$L_{M_2}$}  &
	\multicolumn{1}{|r|}{0}  & 1 & 1 & $v_c$   & 1 & 0 & 1 & $v_c$    & 1 & 1 & 0 & $v_c$ \\ \hline 
	\multicolumn{1}{|l|}{$L_{M_3}$}  &
	\multicolumn{1}{|r|}{0}  & -1 & 1 & -$v_c$   & -1 & 0 & 1 & -$v_c$    & -1 & 1 & 0 & -$v_c$ \\ \hline 
	\multicolumn{1}{|l|}{$L_{M_4}$}  &
	\multicolumn{1}{|r|}{0}  & 1 & -1 & -$v_c$   & 1 & 0 & -1 & -$v_c$    & 1 & -1 & 0 & -$v_c$ \\ \hline 
	\multicolumn{1}{|l|}{$L_{R}$}  &
	\multicolumn{1}{|r|}{6}  & 0 & -1 & 0   & 0 & 6 & -1 & 0    & -1 & -1 & 7 & 0 \\ \hline 
	\end{tabular}
	\end{center}
\end{table*}

The SOLGs can be realized in practice using available electronic devices. Here, we give an example of a universal SO gate that, by changing internal parameters, can work as AND, OR or XOR SO gates. It is universal because using AND and XOR, or OR and XOR we have a complete boolean basis set.

In Fig.~\ref{DMM_figure_6} we give the circuit schematic of the SO universal gate we propose. The logical 0 and 1 are encoded into the potentials $v_1$, $v_2$ and $v_o$ at the terminals. For example, we can choose a reference voltage $v_c$ such that terminals with voltage $v_c$ encode the logic 1s and terminals with voltage $-v_c$ encode logic 0s. The basic circuit elements are (Fig.~\ref{DMM_figure_6}) resistors, memristors (resistors with memory) and voltage-controlled voltage generators. We briefly discuss the last two elements since resistors are straightforward.

The standard equations of a memristor are  given by the following relations \cite{09_memelements}
\begin{align}
&v_M(t)  =M(x)i_M(t) \hspace{.5cm}&  \label{mem1} \\ 
&C\dot v_M(t) =i_C(t) \hspace{.5cm}&  \label{paras_cap} \\ 
&\dot x(t) =f_M(x,v_M)  \hspace{.5cm}& \text{voltage driven} \label{mem2}\\
& \dot x(t)  =f_M(x,i_M)  \hspace{.5cm}& \text{current driven}  \label{mem3}
\end{align}
where $x$ denotes the state variable(s) describing the internal state(s) of the system (from now on we assume for simplicity a single internal variable); $v_M$ and $i_M$ the voltage and current across the memristor. The function $M$ is a monotonous positive function of $x$. In this work we choose the following relation~\cite{08_strukov}
\begin{equation}
M(x) = R_{on}(1-x)+R_{off}x \label{M_eq},
\end{equation}
which is a good approximation to describe the operation of a certain type of memristors. 
This model also includes a small capacitance $C$ in parallel to the memristor that represents parasitic capacitive effects (Eq.~\eqref{paras_cap}).
$f_M$ is a monotonic function of $v_M$ ($i_M$) while $x\in [0,1]$ and null otherwise. We will discuss extensively an actual form of $f_M$ in Sec.~\ref{Math_SOLC_section}. However, any function that satisfies the monotonic condition and nullity for  $x\notin [0,1]$ would define a memristor. Again, we point out here that this particular choice of elements is not unique and indeed it can be accomplished with other types of devices. For instance, we could replace the memristor functionality with an appropriate combination of transistors. 

The voltage-controlled voltage generator (VCVG) is a linear voltage generator piloted by the voltages  $v_1$, $v_2$ and $v_o$. The output voltage is given by 
\begin{equation}
v_{VCVG} = a_1v_1+a_2v_2+a_ov_o+dc, \label{VDVG}
\end{equation}  
and the parameters $a_1$, $a_2$, $a_o$ and $dc$ are determined to satisfy a set of constraints characteristic of each gate (AND, OR or XOR). These constrains can be summerized in the following scheme:
\begin{itemize}
	\item if the gate is connected to a network and the gate configuration is correct, no current flows from any terminal (the gate is in stable equilibrium).
	\item Otherwise, a current of the order of $v_c/R_{on}$ flows with sign opposite to the sign of the voltage at the terminal.
\end{itemize}
When we connect these gates together, these simple requirements induce {\it feedback} in the network forcing it to satisfy, all at once, all the gates since their correct configurations are stable equilibrium points. It is worth noticing that, since all elements but the voltage driven voltage generators are passive elements, the equilibrium points of the gates are {\it stable} and {\it attractive}. 

A set of parameters that satisfies these requirements is given in table \ref{table1} for the three SO gates. Finally, we remark that the parameters in table \ref{table1} are not unique and their choice can strongly affect the dynamics of the network.

%%%%%%%%%%%%%%%%%%%%
\subsection{Auxiliary circuitry for SOLCs}\label{MEM_SOLC_subsection}
%%%%%%%%%%%%%%%%%%%%
\begin{figure}	
	\centerline{\includegraphics[width=.9\columnwidth]{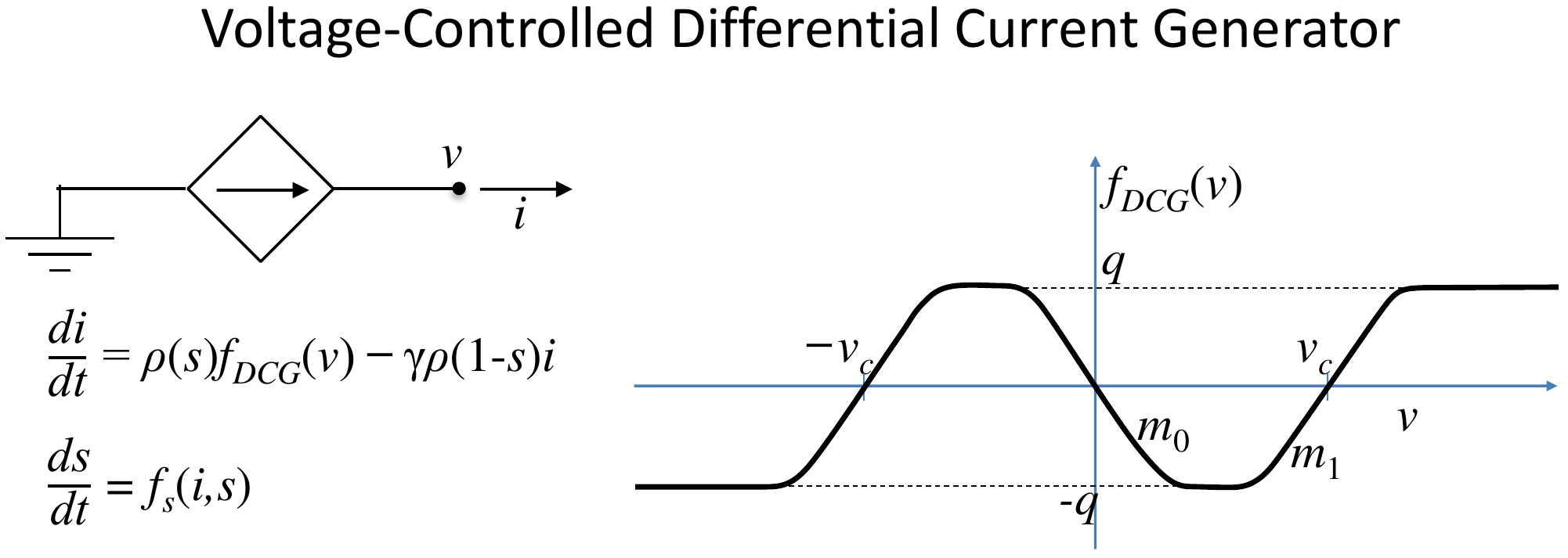}}			
	\caption{\label{DMM_figure_7}Left: circuit symbol and the equations of the Voltage-Controlled Differential Current Generator. Right: sketch of the function $f_{DCG}$.}
\end{figure}
\begin{figure}
\centerline{\includegraphics[width=.8\columnwidth]{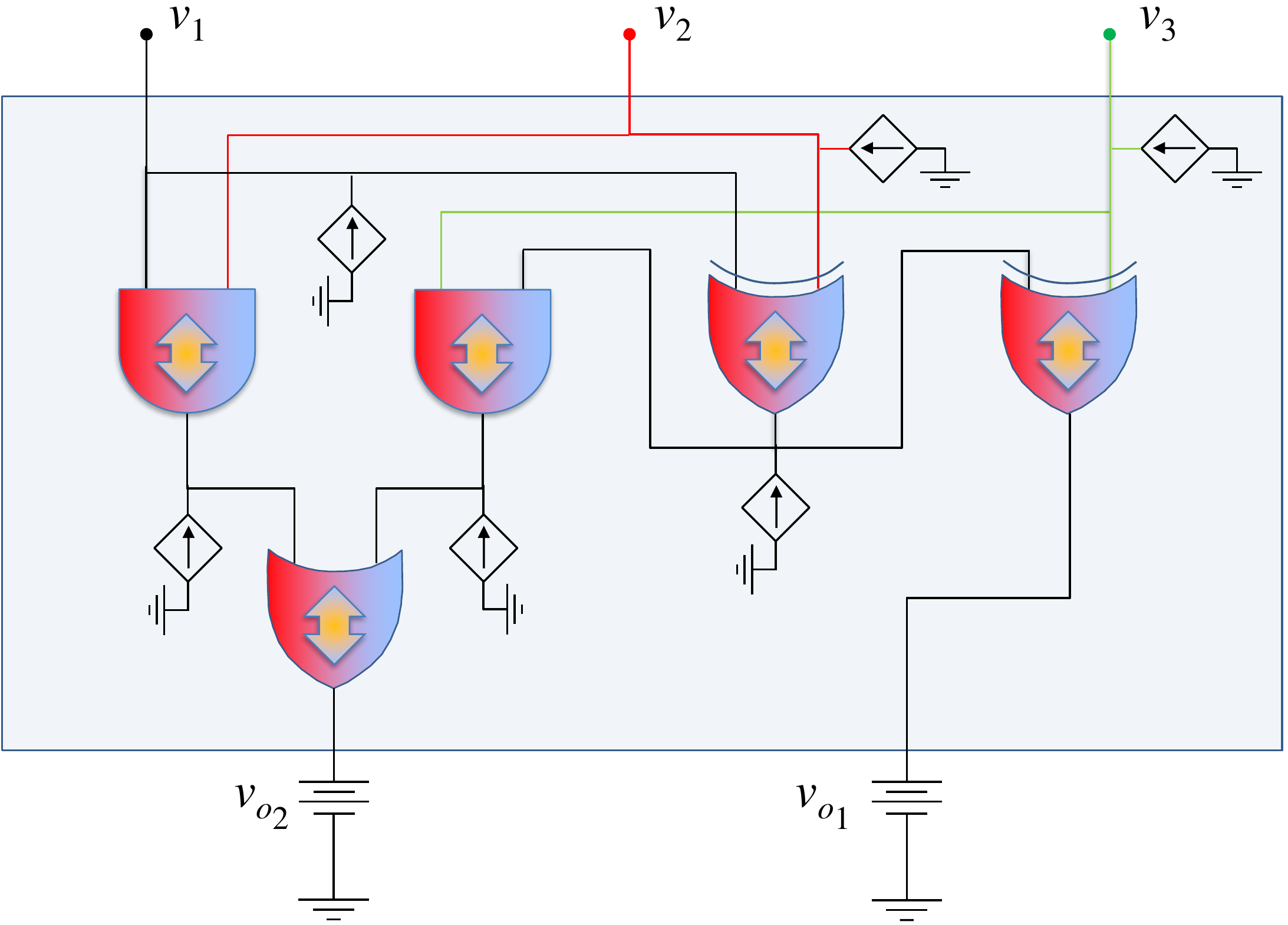}	}	
	\caption{\label{DMM_figure_8}Self-organizing three-bit adder. The input of the adder are the DC generators with voltage $v_{o_1}$ and $v_{o_2}$. The circuit self-organizes to give $v_1$,  $v_2$ and $v_3$ consistent with $v_{o_1}$ and $v_{o_2}$.}
\end{figure}

When we assemble a SOLC using the SOLGs designed in the previous section we may not always prevent the existence of some stable solution not related to the problem. For instance, it can be shown that the DCMs defined through table \ref{table1} admit in some configuration also zero voltage as stable value, namely neither $v_c$ nor $-v_c$. Indeed, if we impose at the output terminal of the SO-AND $v_o=-v_c$ (meaning that we are imposing a logic 0 at the output of the SO-AND), a possible stable solution is $(v_1,v_2)=(0,0)$. (Notice that the only acceptable solutions are $(-v_c,v_c)$, $(v_c,-v_c)$ and $(-v_c,-v_c)$.) 

In order to eliminate this scenario, when we connect together the SOLGs, at each terminal but the ones at which we send the inputs, we connect a Voltage-Controlled Differential Current Generator (VCDCG) sketched in Fig.~\ref{DMM_figure_7}.

By construction, the VCDCG admits as unique stable solutions either $v=v_c$ or  $v=-v_c$. In order to explain how it works, we consider the simplified equation that governs the VCDCG (the complete version will be discussed in detail in section \ref{Math_SOLC_section})
\begin{equation}
\frac{di}{dt}=f_{DCG}(v)\label{VCDG_simplified}
\end{equation}
where the function $f_{DCG}$ is sketched in Fig.~\ref{DMM_figure_7}. If we consider the voltage $v$ around 0, the Eq.~\eqref{VCDG_simplified} can be linearized and gives $\frac{di}{dt}=-m_0v$ with $m_0= -\frac{\partial f_{DCG}}{\partial v}|_{v=0}>0$. Therefore, it is equivalent to a negative inductor and it is sufficient to make the solution 0 unstable. On the other hand, if we linearize around $v=\pm v_c$ (the desired values) we obtain $\frac{di}{dt}=m_1(v\mp v_c)$ with $m_1= \frac{\partial f_{DCG}}{\partial v}|_{v=\pm v_c}>0$. This case is equivalent to an inductor in series with a DC voltage generator of magnitude $\pm v_c$. Since it is connected to a circuit made of memristors and linear voltage generators, $v=\pm v_c$ are stable points. Any other voltage $v$ induces an increase or decrease of the current and, therefore, there are no other possible stable points.

In Fig.~\ref{DMM_figure_8}, we provide an example of how to assemble a SOLC, in particular a SO three-bit adder. At each gate terminal a VCDCG is connected except at the terminal where we send the input. In fact, the inputs are sent through DC voltage generators and, therefore, at those terminals VCDCG would be irrelevant.

%%%%%%%%%%%%%%%%%%%%
\section{Stability analysis of SOLCs}\label{Math_SOLC_section}
%%%%%%%%%%%%%%%%%%%%

In this section we provide a detailed analysis of the SOLC equations, discuss their properties and the stability of the 
associated physical system. In particular, we demonstrate that these SOLCs are {\it asymptotically smooth}, {\it dissipative} and as consequence they have a {\it global attractor}~\cite{hale_2010_asymptotic}. We also prove that all the possible stable equilibria are related to the solution of the CB problem we are implementing, and all the orbits converge exponentially fast to these equilibria {\it irrespective} of the initial conditions. In addition, we prove that this convergence rate depends at most polynomially with the size of the SOLC. Finally, we discuss the absence of periodic orbits and strange attractors in the global attractor. 

We reiterate here that the word ``dissipative'' is meant in the functional analysis sense~\cite{hale_2010_asymptotic} and is not synonymous of ``physically passive'': a passive system (in the physical sense) is necessarily dissipative (in the functional analysis sense), but not the other way around. In fact, the systems we consider in this section are active, and yet functionally dissipative. 

%%%%%%%%%%%%%%%%%%%%
\subsection{Circuit Equations}
%%%%%%%%%%%%%%%%%%%%

The SOLC dynamic is described by standard circuit equations. Using the modified nodal analysis it has been shown that they form a differential algebraic system (DAS) \cite{wiley_enc}. However, dealing with ordinary differential equations (ODEs) is generally simpler from a theoretical point of view since it allows us to use several results of functional analysis as applied to 
dynamical systems theory~\cite{hale_2010_asymptotic}. Of course, the results we obtain for the ODEs apply also to the particular DAS the ODEs originate from. 

For this reason we first perform an order reduction of our system as explained in \cite{wiley_enc}, and after reducing the linear part we obtain an ODE. In doing so, the variables we need to describe our ODE are only the voltages $v_M\in\Rbb^{n_M}$ across the memristors ($n_M$ is the number of memristors), the internal variables $x\in\Rbb^{n_M}$ of the memristors, the currents $i_{DCG}\in\Rbb^{n_{DCG}}$ flowing into the VCDCGs ($n_{DCG}$ the number of VCDCGs), and their internal variables $s\in\Rbb^{n_{DCG}}$. The equations can be formally written as
\begin{gather}
	C\frac{d}{dt}v_M=(A_v+B_v\mathcal{D[}g(x)])v_M+A_{i}i_{DCG}+b,\label{linear}\\
	\frac{d}{dt}x=-\alpha\mathcal{D}[h(x,v_M)]\mathcal{D[}g(x)]v_M,	\label{memristor}\\
	\frac{d}{dt}i_{DCG}=\mathcal{D}[\rho(s)]f_{DCG}(v_M)-\gamma\mathcal{D}[\rho(1-s)]i_{DCG},\label{inductor}\\
	\dfrac{d}{dt}s=f_{s}(i_{DCG},s), \label{inductor_internal}%
\end{gather}
where $\mathcal{D}$ is the linear operator $\mathcal{D}[\cdot]=\mathrm{diag}%
[\cdot]$. In compact form it reads%
\begin{equation}
	\frac{d}{dt}\mathbf{x}=F(\mathbf{x}), \label{syst}%
\end{equation}
where $\mathbf{x}\equiv\{v_M,x,i_{DCG}, s\}$, and $F$ can be read from the right-hand side (r.h.s.) of Eqs.~(\ref{linear})-(\ref{inductor_internal}).
We discuss separately each one of them and we give the definitions of all parameters and functions in the following sections.

%%%%%%%%%%%%%%%%%%%%
\subsection{Equation (\ref{linear})\label{eq1_sec}}
%%%%%%%%%%%%%%%%%%%%

In the Eq.~\eqref{linear} $C$ is the parasitic capacitance of the memristors, $A_v,B_v\in\Rbb^{n_M\times n_M}$ and $A_{i}\in\Rbb^{n_M\times n_{DCG}}$ are constant matrices derived form the reduction of the circuit equations to the ODE format \cite{wiley_enc}. Finally, the vector function $g:\Rbb^{n_M}\rightarrow\Rbb^{n_M}$ is the
conductance vector of the memristors defined through its components as (compare against Eq.~\eqref{M_eq})
\begin{equation}
g_j(x)=(R_1x_j+R_{on})^{-1},\label{g}%
\end{equation}
with $R_1=R_{off}-R_{on}$. As we will discuss in depth in Sec.~\ref{eq1_sec}, we can consider only the variable $x$ restricted to $x\in[0,1]^{n_M}\subset\Rbb^{n_M}$. In this case $g:[0,1]^{n_M}\rightarrow\Rbb^{n_M}$ belongs to $C^{\infty}$ in the whole domain since
$R_{1},R_{on}>0$.

The constant vector $b$ is a linear transformation of the DC components of the
voltage generators. Therefore, we have the following proposition:
\begin{prop}
	The vector function
	\begin{equation}
	F_1(\xvec)=C^{-1}(A_v+B_v\mathcal{D}[g(x)])v_M+C^{-1}A_i i_{DCG}+C^{-1}b \label{F1}
	\end{equation}
	on the r.h.s. of  Eq.~(\ref{linear}) defined on $\xvec\in\Rbb^{n_M}\times[0,1]^{n_M}\times\Rbb^{2n_{DCG}}$	belongs to $C^{\infty}(\xvec)$. Moreover, since the system with no VCDCGs (i.e., $i_{DCG}=0$) is passive the eigenvalues $\lambda_j(x)$ of the matrix $A_v+B_v\mathcal{D}[g(x)]$ satisfy%
	\begin{equation}
	\operatorname{Re}(\lambda_j(x))<0\text{ \ \ for }j=1,\cdots,n_M\text{ and }\forall x\in\lbrack0,1]^{n_M}\text{.} \label{eig1}%
	\end{equation}	
\end{prop}

%%%%%%%%%%%%%%%%%%%%
\subsection{Equation (\ref{memristor})\label{eq2_sec}}
%%%%%%%%%%%%%%%%%%%%

We first remark that the component $j$-th on the r.h.s. of Eq.~\eqref{memristor} depends only on the variables $x_j$ and $v_{M_j}$ and we can write for $j=1,\cdots,n_M$
\begin{equation}
\frac{d}{dt}x_j=-\alpha h_j(x,v_M)g_j(x)v_{M_j}=-\alpha h(x_j,v_{M_j})g(x_j)v_{M_j}\,.\label{memristor1}
\end{equation}
Therefore, without loss of generality, we can discuss  Eq.~\eqref{memristor1} suppressing the index $j$ in place of Eq.~\eqref{memristor}.

The \eqref{memristor1} is the equation for the current-driven memristors with $g(x)v_M$ the current flowing through the memristor. The case of voltage-driven does not add any extra content to our discussion and all models and results can be similarly derived for the voltage-driven case. The coefficient $\alpha>0$ can be linked to the physics of the memristors \cite{13_properties}. The conductance $g(x)$ is given by \eqref{g} and it has been discussed in Sec.~\ref{eq1_sec}.

The nonlinear function $h(x,v_M)$ can be chosen in many ways (see e.g., \cite{pershin11a, memrisstor_models} and references therein), and a compact and useful way to express it (from a numerical point of view) is~\cite{pershin11a}
\begin{equation}
h(x,v_M)=\theta(x)\theta(v_M)+\theta(1-x)\theta(-v_M),
\end{equation}
where $\theta$ is the Heaviside step function. However, this function does not belong to any class of continuity (like the majority of the alternative models) for the variables $x$ and $v_M$. We then change it in a way that is physically consistent and it belongs to some class of continuity. In order to do this we can write
\begin{equation}
h(x,v_M)=\left(  1-e^{-kx}\right)  \tilde{\theta}(v_M)+\left(1-e^{-k(1-x)}\right)  \tilde{\theta}(-v_M). \label{h1}
\end{equation}
Before giving a complete description of the function $\tilde{\theta}$ we only require at the moment that
\begin{equation}
\tilde{\theta}(y)=\left\{
\begin{array}[c]{cc}
>0 & \text{for }y>0\\
0   & \text{for }y\leq0.
\end{array}
\right.  \label{theta1}
\end{equation}
Using (\ref{h1}) and \eqref{theta1}, and for $k\gg1$ (that is a limit physically consistent) the function $-\alpha h(x,v_M)g(x)v_M$ can be linearized around $x=0$ and $x=1$ in the following way
\begin{align}
-&\alpha h(x,v_M)g(x)v_M|_{x\approx 0}\nonumber\\
&=\left\{
\begin{array}[c]{cc}
-\alpha kx\tilde{\theta}(v_M)g(0)|v_M|  & \text{for }v_M>0\\
\alpha\tilde{\theta}(-v_M)g(0)|v_M| & \text{for }v_M\leq0,
\end{array}
\right. \\
-&\alpha h(x,v_M)g(x)v_M|_{x\approx 1}\nonumber\\
&=\left\{
\begin{array}[c]{cc}
-\alpha\tilde{\theta}(v_M)g(1)|v_M| & \text{for }v_M>0\\
-\alpha k\left(  x-1\right)  \tilde{\theta}(-v_M)g(1)|v_M| & \text{for }v_M\leq0.
\end{array}
\right.
\end{align}
From these expressions and using \eqref{memristor1} the behavior of $x(t)$ around either $0$ or $1$ can be evaluated and we have
\begin{align}
x(t)|_{x\approx0}  &  \approx\left\{
\begin{array}[c]{cc}
e^{-\alpha k\tilde{\theta}(v_M)g(0)|v_M|t} & \text{for }v_M>0\\
\alpha\tilde{\theta}(-v_M)g(0)|v_M|t & \text{for }v_M\leq 0,
\end{array}
\right. \\
x(t)|_{x\approx1}  &  \approx\left\{
\begin{array}[c]{cc}
1-\alpha\tilde{\theta}(v_M)g(1)|v_M|t & \text{for }v_M>0\\
1-e^{-\alpha k\tilde{\theta}(-v_M)g(0)|v_M|t} &\text{for }v_M\leq0.
\end{array}
\right.
\end{align}
This proves the following proposition:

\begin{prop}
	Using \eqref{h1} to represent $h(x,v_M)$ then $h(x,v_M)\in C^{\infty}(x)$ and for any $x\in[0,1]^{n_M}$ and for any $t>0$, we have $\phi_t^x(x)\in[0,1]^{n_M}$ were $\phi_t^x(x)$ is the flow of the dynamical system \eqref{syst} restricted to the variable $x$ only. Then $[0,1]^{n_M}$ is an invariant subset of $\Rbb^{n_M}$ under
	$\phi_t^x$. Moreover, the boundary points $0$ and $1$ are limit points and for any open ball of $B\subset[0,1]^{n_M}$ we have that $\phi_t^x(B)\subset[0,1]^{n_M}$ is an open ball.
\end{prop}

Since, physically speaking, $x$ should be restricted to $[0,1]^{n_M}$, this proposition
allows us to restrict the values of $x$ to $[0,1]^{n_M}$ in a natural way by using a function $h(x,v_M)\in C^{\infty}(x)$.

Now, we discuss an actual expression of $\tilde{\theta}(y)$ satisfying the condition \eqref{theta1} and other useful conditions for the next sections. The goal is to find a $\tilde{\theta}(y)$ that satisfies the following conditions

\begin{enumerate}
	\item $\tilde{\theta}(y)$ satisfies \eqref{theta1},
    \item \label{condition_2}$\tilde{\theta}(y)=0$ for any $y\leq0$,
    \item $\tilde{\theta}(y)=1$ for any $y\geq1$,
    \item \label{condition_4}for some $r\geq1$ and for $l=1,\cdots,r$ the derivatives $\left.  \frac{d^{l}}{dy^{l}}\tilde{\theta}(y)\right\vert _{y=0}=\left.
	\frac{d^{l}}{dy^{l}}\tilde{\theta}(y)\right\vert _{y=1}=0$.
\end{enumerate}

\begin{figure}	
	\centerline{\includegraphics[width=1\columnwidth]{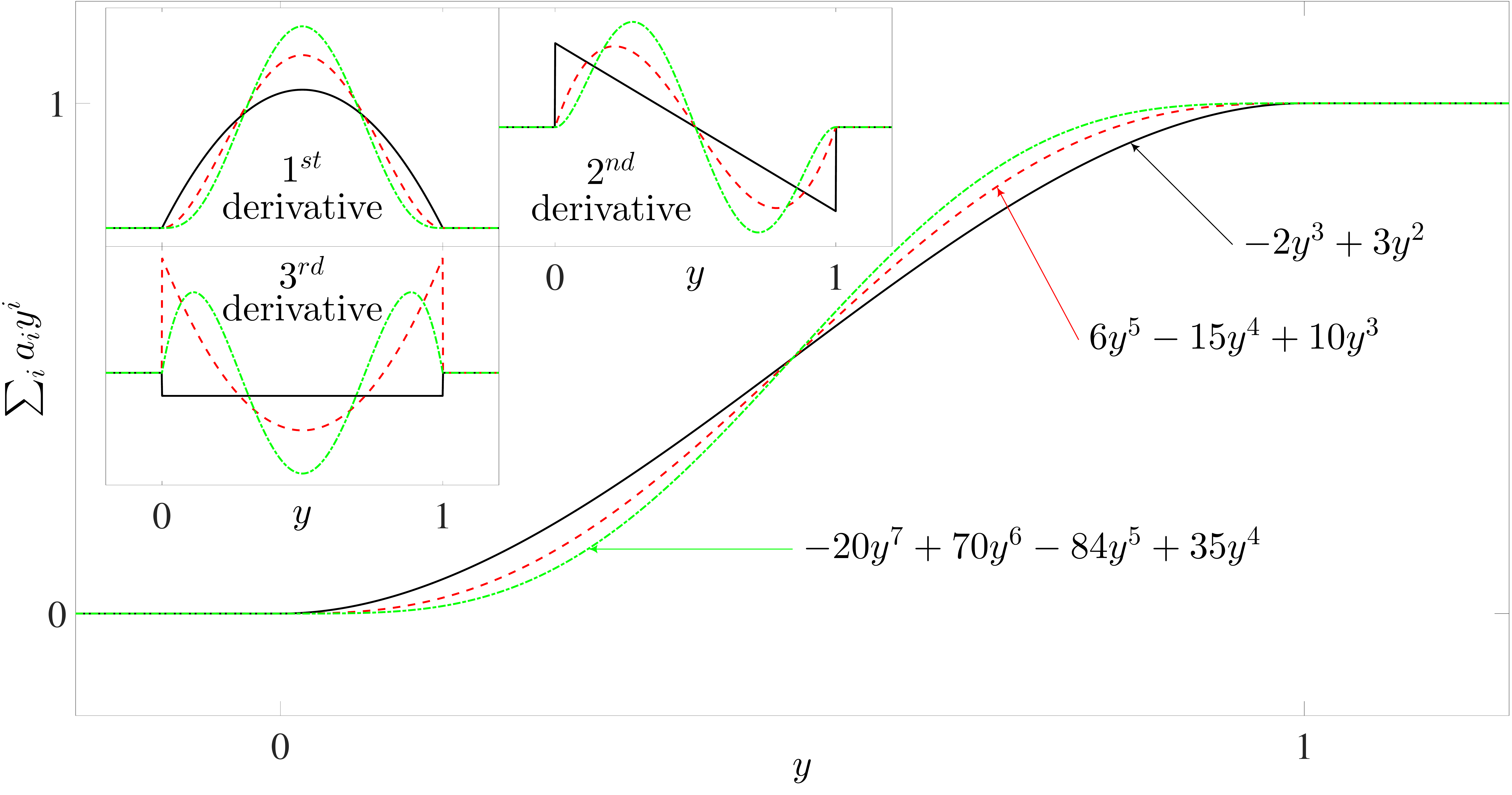}}			
	\caption{\label{figure_theta}Examples of function $\tilde{\theta}^r(y)$ [Eq.~(\ref{theta2})] for $r=1$ (black solid line), $r=2$ (red dashed line) and $r=3$ (green dashed-dot line). In the insets the 1$^{st}$, 2$^{nd}$ and 3$^{rd}$ derivative are plotted.}	
\end{figure}

It is easy to prove that the conditions \ref{condition_2}-\ref{condition_4} are
satisfied by
\begin{equation}
\tilde{\theta}^{r}(y)=\left\{
\begin{array}[c]{ll}
1 & \text{for }y>1\\
\sum_{i=r+1}^{2r+1}a_iy^i & \text{for }0\leq y\leq1\\
0 & \text{for }y<0
\end{array}
\right.  \label{theta2}
\end{equation}
where the coefficients $a_i$ with $i=r+1,\cdots,2r+1$ can be evaluated by requiring $\sum_{i=r+1}^{2r+1}a_i=1$ and $\sum_{i=r+1}^{2r+1}\binom{i}{l}a_i=0$ for $l=1,\cdots,r$. From the conditions \ref{condition_2}-\ref{condition_4} the polynomial $\sum_{i=r+1}^{2r+1}a_iy^i$ is also monotonous because the stationary points are only located on $0$ and $1$ (because of the condition \ref{condition_4}). Therefore the polynomial
$\sum_{i=r+1}^{2r+1}a_iy^i$ satisfying conditions \ref{condition_2}-\ref{condition_4} also satisfies \eqref{theta1}.

We also can notice that $\lim_{r\rightarrow\infty}\tilde{\theta}^r(y)=\theta(y-\tfrac{1}{2})$ and it is easy to show that a compact integral representation of $\sum_{i=r+1}^{2r+1}a_iy^i$ exists and reads \begin{equation}
\sum_{i=r+1}^{2r+1}a_iy^i{=}\left(\int_{0}^{1}z^{r+1}(z{-}1)^{r+1}dz\right)^{-1}\hspace{-8pt}\int_{0}^{y}z^{r+1}(z{-}1)^{r+1}dz.
\end{equation}
As an illustration we report the cases of $r=$1, 2 and 3 in Fig.~\ref{figure_theta}.%

In conclusion, Eq.~(\ref{theta2}) allows us to write the following proposition:

\begin{prop}
	Using \eqref{h1} to represent $h(x,v_M)$ and \eqref{theta2} to represent $\tilde{\theta}(v_M)$, then, for any particular choice of $r$, we have that the function
	\begin{equation}
	F_{2}(\mathbf{x})=-\alpha\mathcal{D[}h(x,v_M)]\mathcal{D[}g(x)]v_M%
	\end{equation}
	on the r.h.s. of \eqref{memristor} defined on $\xvec\in\Rbb^{n_M}\times[0,1]^{n_M}\times\Rbb^{2n_{DCG}}$ is in the class $C^{r}(\xvec)$.
\end{prop}

We finally briefly discuss an interesting physical consequence of modelling $h(x,v_M)$ using $\tilde{\theta}^r(y)$. In fact, if we consider
\begin{equation}
h(x,v_M)=\left(1-e^{-kx}\right)  \tilde{\theta}^r(\tfrac{v_M}{2V_t})+\left(  1-e^{-k(1-x)}\right)  \tilde{\theta}^r(-\tfrac{v_M}{2V_t}),\label{h2}
\end{equation}
we can also interpret $V_t$ as the threshold of the memristor \cite{13_properties}, thus the $\tilde{\theta}^r$ enables a natural and smooth way to include threshold effects in the memristor equations.

%%%%%%%%%%%%%%%%%%%%
\subsection{Equation (\ref{inductor})\label{eq3_sec}}
%%%%%%%%%%%%%%%%%%%%

Here, we have $\gamma>0$. Each component of the vector function $f_{DCG}:\Rbb^{n_M}\rightarrow\Rbb^{n_{DCG}}$ actually depends on the voltage at the node where the VCDCG is connected. This voltage is expressed as a linear combination of the $v_M$ components
\begin{equation}
v_{DCG_j}=u_j^Tv_M+v_{0_j} \label{vDCG}%
\end{equation}
where $u_j\in\Rbb^{n_M}$ and $v_{0_j}\in\Rbb$ are constant vectors. Therefore, we can write
\begin{equation}
f_{DCG_j}(v_M)=f_{DCG_j}(v_{DCG_j})=f_{DCG}(v_{DCG_j})\text{.}
\end{equation}

The function we want to reproduce with our VCDCG is depicted in  Fig.~\ref{DMM_figure_7}. This shape can be obtained in several ways using several smooth step functions like $\operatorname{erf}$, $\arctan$ or even the $\tilde{\theta}^r$ we defined in the previous section. Therefore, we can assume that $f_{DCG_j}$ is at least of $C^r(v_{DCG})$.

Finally, the function $\rho$ also satisfies
\begin{equation}
\rho_j(s)=\rho_j(s_j)=\rho(s_j),
\end{equation}
and is expressed as
\begin{equation}
\rho(s_j)=\tilde{\theta}^r\left(\tfrac{s_j-\frac{1}{2}}{\delta s}\right)\label{ro},
\end{equation}
with $0<\delta s\ll 1$. Therefore, we have the following proposition:

\begin{prop}
	The vector function
	\begin{equation}
	F_{3}(\xvec)=\mathcal{D}[\rho(s)]f_{DCG}(v_M)-\gamma\mathcal{D}[\rho(1-s)]i_{DCG}
	\end{equation}
	on the r.h.s. of \eqref{inductor} defined on $\xvec\in\Rbb^{n_M}\times[0,1]^{n_M}\times\Rbb^{2n_{DCG}}$ is at least of $C^r(\xvec)$.
\end{prop}

%%%%%%%%%%%%%%%%%%%%
\subsection{Equation (\ref{inductor_internal})\label{eq4_sec}}
%%%%%%%%%%%%%%%%%%%%

\begin{figure}	
	\centerline{\includegraphics[width=.9\columnwidth]{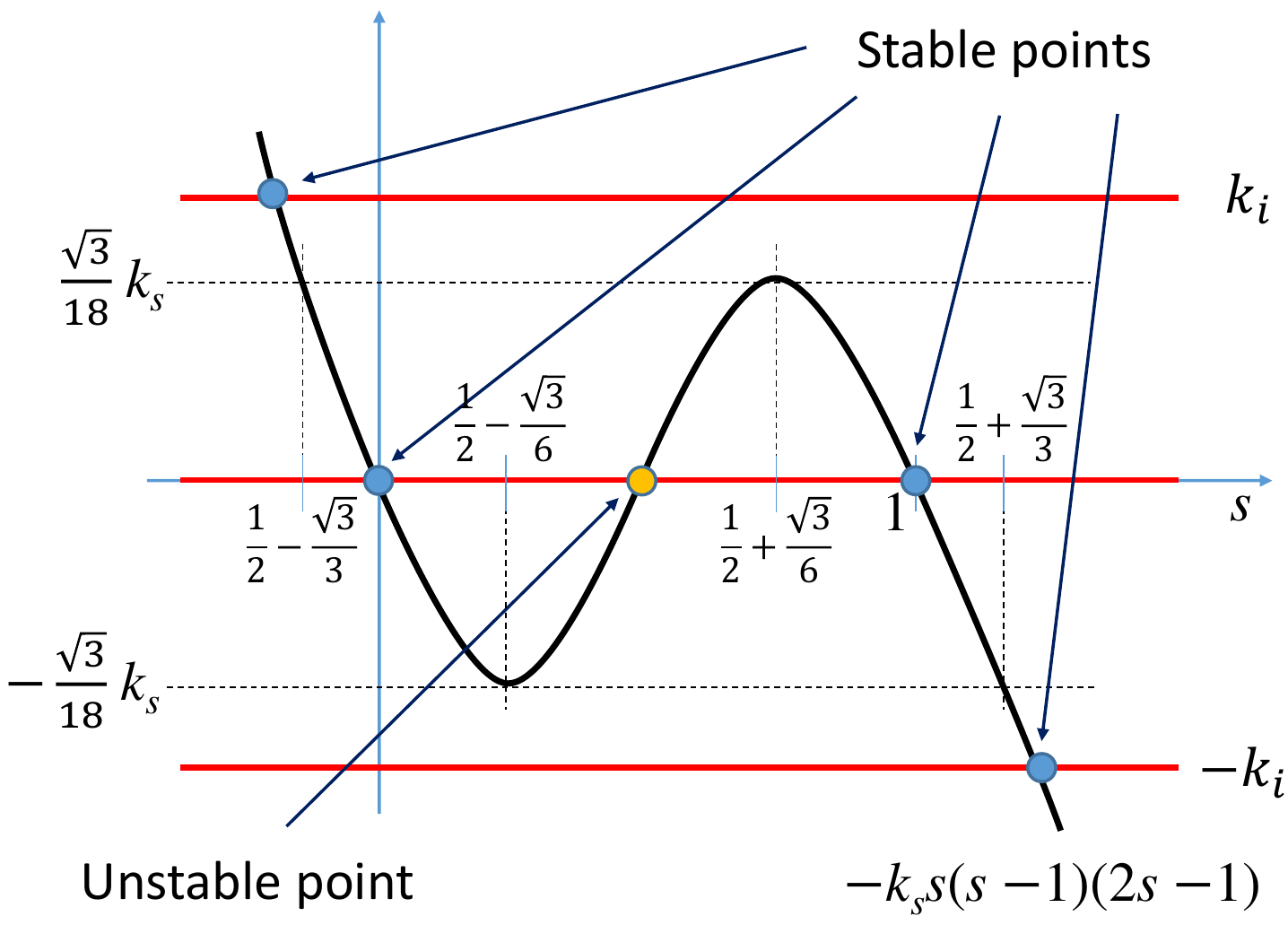}}			
	\caption{\label{stable_fs}Stability picture of  Eq.~\eqref{inductor_internal}.}	
\end{figure}

We finally discuss the r.h.s. of Eq.~\eqref{inductor_internal}. Here, we have the function $f_s:\Rbb^{n_{DCG}}\times\Rbb^{n_{DCG}}\rightarrow\Rbb^{n_{DCG}}$ and it satisfies component wise
\begin{equation}
f_{s_j}(i_{DCG},s)=f_{s_j}(i_{DCG},s_j)=f_s(i_{DCG},s_j).
\end{equation}
We can then discuss it as a scalar function of two variables neglecting the index $j$.

The function $f_s(i_{DCG},s)$ is 
\begin{align}
f_s(&i_{DCG},s)=-k_ss(s-1)(2s-1)+\nonumber\\
&k_i\left(1-\prod_j \tilde\theta\left(\tfrac{i_{\min}^2-i_{DCG_j}^2}{\delta i}\right)-
\prod_j \tilde\theta\left(\tfrac{i_{\max}^2-i_{DCG_j}^2}{\delta i}\right)\right)
\label{fs}
\end{align}
with $k_s,k_i,\delta i,i_{\min},i_{\max}>0$ and $i_{\min}<i_{\max}$. Note that when $k_i=0$ Eq.~\eqref{fs} represents a bistable system. To understand the role of the variable $s$ we notice that, by considering only the term in $s$ in \eqref{fs}, it represents the bistable system with two stable equilibrium points in $s=0,1$ and an unstable equilibrium in $s=\frac{1}{2}$. Now, we consider again the terms in $i_{DCG}$ and $\delta i\ll i_{\min}$. In this case,  $\prod_j \tilde\theta((i_{\min}^2-i_{DCG_j}^2)/\delta i)=0$ if at least one $i_{DCG_j}$ satisfies $|i_{DCG_j}|>i_{\min}+\delta i/2\approx i_{\min}$, otherwise the product is 1. On the other hand, we have $\prod_j \tilde\theta((i_{\max}^2-i_{DCG_j}^2)/\delta i)=0$ if at least one $i_{DCG_j}$ satisfies $|i_{DCG_j}|>i_{\max}+\delta i/2\approx i_{\max}$, otherwise the product is 1. Therefore, if we consider $k_i>\sqrt{3/18}k_s$, we have the stability picture described in Fig.~\ref{stable_fs} where 
\begin{itemize}
	\item the red line located at $k_i$ represents terms in $i_{DCG}$ of \eqref{fs} for the case in which at least one $i_{DCG_j}$ satisfies $|i_{DCG_j}|>i_{\max}$. Therefore, we have only one stable equilibrium.
	\item The red line located at 0 is for the case in which at least one $i_{DCG_j}$ satisfies $|i_{DCG_j}|>i_{\min}$ and all $i_{DCG_j}$ satisfy $|i_{DCG_j}|<i_{\max}$. Therefore, we have two possible stable equilibria and an unstable one. 
	\item The red line located at $-k_i$ is for the case in which all $i_{DCG_j}$ satisfy $|i_{DCG_j}|<i_{\min}$. Therefore, we have only one stable equilibrium. 
\end{itemize}

This picture can be summarized as follows: if at least one $|i_{DCG_j}|>|i_{\max}|$ then the variable $s$ will approach the {\it unique} stable point for $s<\frac{1}{2}-\frac{\sqrt{3}}{3}<0$, while if all $|i_{DCG_j}|<|i_{DCG_{\min} }|$ the variable $s$ will approach the {\it unique} stable point for $s>\frac{1}{2}+\frac{\sqrt{3}}{3}>1$. If at least one $|i_{DCG_j}|>i_{\min}$ and all $|i_{DCG_j}|<i_{\max}$ then $s$ will be either $\frac{1}{2}-\frac{\sqrt{3}}{3}<s<\frac{1}{2}-\frac{\sqrt{3}}{6}$ or $\frac{1}{2}+\frac{\sqrt{3}}{6}<s<\frac{1}{2}+\frac{\sqrt{3}}{3}$.

Now, from Eq.~\eqref{inductor} and the definition \eqref{ro} we have that, if at least one $|i_{DCG_j}|>|i_{\max}|$, then $s<0$, $\rho(s)=0$ and $\rho(1-s)=1$ and the equation \eqref{inductor} reduces to $\frac{d}{dt}i_{DCG}=-\gamma i_{DCG}$. Therefore, the currents in $i_{DCG}$ decrease. 

When all of them reach a value $|i_{DCG_j}|<|i_{\min}|$ then $s$ jumps and reaches the stable point $s<1$ and we have $\rho(s)=1$ and $\rho(1-s)=0$ and the Eq.~\eqref{inductor} reduces to $\frac{d}{dt}i_{DCG}=f_{DCG}(v_M)$ (which is the simplified version \eqref{VCDG_simplified} discussed in Sec.~\ref{MEM_SOLC_subsection}). Since $f_{DCG}(v_M)$ is bounded, the current $i_{DCG}$ is bounded and, if $k_s\gg\max(|f_{DCG}(v_M)|)$, we have 
\begin{equation}
\sup(|i_{DCG}|)\simeq|i_{\max}|.
\end{equation}

Therefore, using the analysis in this section we can conclude with the following proposition:

\begin{prop}
	\label{proposition_F4}The vector function 
	\begin{align}
	F_4&(\xvec)=-k_ss(s-1)(2s-1)+\nonumber\\
&k_i\left(1-\prod_j \tilde\theta\left(\tfrac{i_{\min}^2-i_{DCG_j}^2}{\delta i}\right)-
\prod_j \tilde\theta\left(\tfrac{i_{\max}^2-i_{DCG_j}^2}{\delta i}\right)\right)
	\end{align}
	on the r.h.s. of \eqref{inductor} defined on $\xvec\in\Rbb^{n_M}\times[0,1]^{n_M}\times\Rbb^{2n_{DCG}}$ is at least of  $C^{r}(\xvec)$. Moreover, there exist $i_{\max},s_{\max}<\infty$ and $s_{\min}>-\infty$ such that the variables $s$ and $i_{DCG}$ can be restricted to $[-i_{\max},i_{\max}]^{n_{DCG}}\times[s_{\min},s_{\max}]^{n_{DCG}}\subset\Rbb^{2n_{DCG}}$. 
	In particular, $s_{\max}$ is the unique zero of $F_4(s,i_{DCG}=0)$. This subset is invariant under the flow $\phi_t^{i_{DCG},s}$ of the dynamical system \eqref{syst} restricted to the variable $i_{DCG}$ and $s$ only. Moreover, the boundary points are limit
	points and for any open ball of $B\subset[-i_{\max},i_{\max}]^{n_{DCG}}\times[s_{\min},s_{\max}]^{n_{DCG}}$ we have that $\phi_t^{i_{DCG},s}(B)\subset[-i_{\max},i_{\max}]^{n_{DCG}}\times[s_{\min},s_{\max}]^{n_{DCG}}$ is an open ball.
\end{prop}

We can conclude this section with the last proposition that can be easily
proven by using \eqref{eig1} and the proposition \ref{proposition_F4}:

\begin{prop}
	If the variables $x\in[0,1]^{n_M}$, $i_{DCG}\in[-i_{\max},i_{\max}]^{n_{DCG}}$ and $s\in[s_{\min},s_{\max}]^{n_{DCG}}$ then there exist $v_{\max}<\infty$ and $v_{\min}>-\infty$ such that $[v_{\min},v_{\max}]^{n_M}\subset\Rbb^{n_M}$ is an invariant subset under the flow $\phi_t^{v_M}$ of the dynamical system \eqref{syst} restricted to the variable $v_M$ only.
\end{prop}

%%%%%%%%%%%%%%%%%%%%
\subsection{Existence of a global attractor \label{global_attractor_section}}
%%%%%%%%%%%%%%%%%%%%

In the previous sections we have presented the equations that govern the dynamics of our SOLCs. We are now ready to perform the stability analysis of the SOLCs. 

We consider our system given by \eqref{linear}-\eqref{inductor_internal}, or in compact form by \eqref{syst}. The terminology we will use here follows the one in Ref.~\cite{hale_2010_asymptotic}, in particular in Chapter 3. For our dynamical system \eqref{syst} we can formally define the semigroup $T(t)$ such that $T(t)\xvec=\xvec+\int_0^tF(T(t^{\prime})\xvec)dt^{\prime}$, or defining $\xvec(t)=T(t)\xvec$ we have the more familiar
\begin{equation}
T(t)\xvec(0)=\xvec(0)+\int_0^tF(\xvec(t^{\prime}))dt^{\prime} \label{cr_semigroup}.
\end{equation}

Since we have proven in the previous sections that $F(\xvec)\in C^{r}(X)$ with $X$ the complete and compact metric space $X=[v_{\min},v_{\max}]^{n_M}\times\lbrack0,1]^{n_M}\times[-i_{\max},i_{\max}]^{n_{DCG}%
}\times\lbrack s_{\min},s_{\max}]^{n_{DCG}}$, then $T(t)$ is a $C^{r}$-semigroup.

We recall now that a $C^{r}$-semigroup is \textit{asymptotically smooth} if, for any nonempty, closed, bounded set $B\subset X$ for which $T(t)B\subset B$, there is a compact set $J\subset B$ such that $J$ attracts $B$ \cite{hale_2010_asymptotic}. Here, the term "attract" is formally defined as: a set $B\subset X$ is said to \textit{attract} a set $C\subset X$ under $T(t)$ if $\textrm{dist}(T(t)C,B)\rightarrow 0$ as $t\rightarrow\infty$ \cite{hale_2010_asymptotic}. 

Moreover, we have that a semigroup $T(t)$ is said to be {\it point dissipative} (compact dissipative) (locally compact dissipative) (bounded dissipative) if there is a bounded set $B\subset X$ that attracts each point of $X$ (each compact set of $X$) (a neighborhood of each compact set of $X$) (each bounded set of $X$) under $T(t)$ \cite{hale_2010_asymptotic}. If $T(t)$ is point, compact, locally, dissipative and bounded dissipative we simply say that $T(t)$ is \textit{dissipative}. Now we are ready to prove the following lemmas:
\begin{lem}\label{asym_smooth}
	The $C^{r}$-semigroup $T(t)$ defined in \eqref{cr_semigroup} is asymptotically smooth.
\end{lem}

\begin{proof}
	In order to prove this lemma we first need to decompose $T(t)$ as the sum $T(t)=S(t)+U(t)$. We take initially $S(t)$ and $U(t)$ as
	\begin{align}
	U&(t)\xvec(0)=\nonumber\\
	&\left(
	\begin{array}[c]{c}
	v_M(0)-k_{DCG}\mathcal{U}^{+}i_{DCG}+\mathcal{U}^{+}v_{0}+\int_{0}^{t}F_{1}(\xvec(t^{\prime}))dt^{\prime}\\
	x(0)+\int_{0}^{t}F_{2}(\xvec(t^{\prime}))dt^{\prime}\\	0\\	0
	\end{array}
	\right) \label{U}\\
	S&(t)\xvec(0)=\left(
	\begin{array}[c]{c}
	k_{DCG}\mathcal{U}^{+}i_{DCG}-\mathcal{U}^{+}v_{0}\\
	0\\ 
	i_{DCG}(0)+\int_{0}^{t}F_{3}(\xvec(t^{\prime}))dt^{\prime}\\ 
	s(0)+\int_{0}^{t}F_{4}(\xvec(t^{\prime}))dt^{\prime}
	\end{array}
	\right)\label{S}
	\end{align}
	where $k_{DCG}>0$, $v_0$ is the the constant vector whose components are the $v_{0_j}$ in \eqref{vDCG}, and $\mathcal{U}^+$ is the pseudoinverse of the matrix $\mathcal{U}$ whose rows are the vectors $u_j^T$ in \eqref{vDCG}.	It is worth noticing that for a well defined SOLC it is easy to show that $\mathcal{UU}^+=I$ (the inverse, i.e. $\mathcal{U}^+\mathcal{U}=I$, does not generally hold).
	
	We also perform two variable shifts:
	\begin{align}
	v_M  &  \rightarrow v_M+\mathcal{U}^{+}v_{0}\label{vchange}\\
	s  &  \rightarrow s-s_{\max} \label{schange}
	\end{align}
	Since they are just shifts they do not formally change anything in  Eqs.~\eqref{U} and \eqref{S}, except for additive terms in $\mathcal{U}^+v_0$ and	$s_{\max}$. Also, the metric space changes accordingly to  $X\rightarrow[v_{\min}+\mathcal{U}^+v_0,v_{\max}+\mathcal{U}^+v_0]^{n_M}\times[0,1]^{n_M}\times[-i_{\max},i_{\max}]^{n_{DCG}}\times[s_{\min}-s_{\max},0]^{n_{DCG}}$. To avoid increasing the burden of notation, in the following we will refer to all variables and operators with the same previous symbols, while keeping in mind the changes \eqref{vchange} and \eqref{schange}.
	
	Now, by definition, $U(t):X\rightarrow[v_{\min}+\mathcal{U}^+v_0,v_{\max}+\mathcal{U}^+v_0]^{n_M}\times[0,1]^{n_M}\times[0,0]^{n_{DCG}}\times[-s_{\max},-s_{\max}]^{n_{DCG}}$ and it is easy to show that it is equivalent to the system
	\begin{align}
	    &C\frac{d}{dt}v_M=(A_v+B_v\mathcal{D}[g(x)])(v_M-\mathcal{U}^+v_0)+b\\
		&\frac{d}{dt}x=-\alpha\mathcal{D}[h(x,v_M)]\mathcal{D}[g(x)](v_M-\mathcal{U}^+v_0)\\
		&i_{DCG}=0\\
		&s=-s_{\max}.
	\end{align}

By construction, from Eq.~\eqref{eig1} and the definition of $h(x,v_M)$ in \eqref{h2}, $U(t)$ represents a {\it globally passive circuit}. It is then asymptotically smooth, completely continuous\footnote{A semigroup $T(t)$, $t>0$, is said to be conditionally completely continuous for $t\geq t_1$ if, for each $t\geq t_1$ and each bounded set $B$ in $X$ for which $\{T(s)B,0\leq s\leq t\}$ is bounded, we have $T(t)B$ precompact. A semigroup $T(t)$, $t\geq0$, is completely continuous if it is conditionally completely 	continuous and, for each $t\geq 0$, the set $\{T(s)B,0\leq s\leq t\}$ is bounded if $B$ is bounded \cite{hale_2010_asymptotic}.}, and since it is defined in a compact metric space $X$, it is dissipative.
	
	Now, following the lemma 3.2.3 of \cite{hale_2010_asymptotic} we only need to prove that there is a continuous function $k:\Rbb^+\rightarrow\Rbb^+$ such that $k(t,r)\rightarrow0$ as $t\rightarrow\infty$ and $|S(t)\xvec|<k(t,r)$ if $|x|<r$. In order to prove this statement, we first see that $S(t)$ is equivalent to the system
	\begin{align}
		&v_M=k_{DCG}\mathcal{U}^+i_{DCG}\\
		&x=0\\
		&\frac{d}{dt}i_{DCG}=\mathcal{D}[\rho(s+s_{\max})]f_{DCG}(v_M)+\nonumber\\
		&\hspace{3cm} -\gamma\mathcal{D}[\rho(1-s-s_{\max})]i_{DCG}\\
		&\dfrac{d}{dt}s=f_{s}(i_{DCG},s).
	\end{align}
	Since $v_M=k_{DCG}\mathcal{U}^+i_{DCG}$ we have that $f_{DCG}(v_M)=f_{DCG}(k_{DCG}\mathcal{U}^+i_{DCG})$. From the definition and discussion on $f_{DCG}$ given in section \ref{eq3_sec}, the variable change \eqref{vDCG} and the definition of $\mathcal{U}^+$, we have that $f_{DCG_j}(v_M)=f_{DCG_j}(v_{DCG_j})=f_{DCG}(k_{DCG}i_{DCG_j})$. Now, since we consider $k_{DCG}$ such that $k_{DCG}i_{\max}<v_c/2$, from the discussion in section \ref{eq4_sec} and considering the variable change \eqref{schange}, the unique stable equilibrium point for $S(t)$ is $\xvec=0$, and it is also a  global attractor in $X$. Moreover, this equilibrium is hyperbolic (see section \ref{equilibria_subsection}), then there is a constant $\xi>0$ such that $|S(t)\xvec|<e^{-\xi t}$. This concludes the proof.
\end{proof}

\begin{lem}
	\label{dissipative_lemma}The $C^r$-semigroup $T(t)$ defined in \eqref{cr_semigroup} is dissipative.
\end{lem}

\begin{proof}
	From lemma \ref{asym_smooth} $T(t)$ is asymptotically smooth, then from the	corollary 3.4.3 of \cite{hale_2010_asymptotic} there is a compact set which attracts compact sets of $X$. There is also a compact set which attracts a	neighborhood of each compact set of $X$. Therefore, since our $X$ is bounded, the lemma follows.
\end{proof}

At this point we recall some other definitions and results from topology of dynamical systems that can be found, e.g., in Ref.~\cite{hale_2010_asymptotic} and will be useful for the next discussions.

For any set $B\subset X$, we define the $\omega$\textit{-limit set} $\omega(B)$ of $B$ as $\omega(B)={\textstyle\bigcap_{s\geq0}}\mathrm{Cl}{\textstyle\bigcup_{t\geq s}}T(t)B$. A set $J\subset X$ is said to be \textit{invariant} if, $T(t)J=J$ for $t\geq0$. A compact invariant set $J$ is said to be a \textit{maximal compact invariant} set if every compact invariant set of the semigroup $T(t)$ belongs to $J$. An invariant set $J$ is \textit{stable} if for any neighborhood $V$ of $J$, there is a neighborhood $V^{\prime}\subseteq V$ of $J$ such that $T(t)V^{\prime}\subset V^{\prime}$ for $t\geq0$. An invariant set $J$ \textit{attracts points locally} if there is a neighborhood $W$ of $J$ such
that $J$ attracts points of $W$. The set $J$ is \textit{asymptotically stable} (a.s.) if $J$ is stable and attracts points locally. The set $J$ is \textit{uniformly asymptotically stable} (u.a.s.) if $J$ is stable and attracts a neighborhood of $J$. An invariant set $J$ is said to be a \textit{global attractor} if $J$ is a maximal compact invariant set which
attracts each bounded set $B\subset X$. In particular, $\omega(B)$ is compact and belongs to $J$ and if $J$ is u.a.s. then $J=\bigcup_{B}\omega(B) $.

Now, we are ready to prove the following theorem:

\begin{thm}
	\label{global_attractor_theorem}The $C^r$-semigroup $T(t)$ defined in \eqref{cr_semigroup} possesses an u.a.s. global attractor $A$.
\end{thm}

\begin{proof}
	From the lemmas \ref{asym_smooth} and \ref{dissipative_lemma} we have that	$T(t)$ is asymptotically smooth and dissipative. Moreover, since $X$ is bounded, orbits of bounded sets are bounded and then the theorem follows directly from the theorems 3.4.2 and 3.4.6 of \cite{hale_2010_asymptotic}.
\end{proof}

%%%%%%%%%%%%%%%%%%%%
\subsection{Equilibrium points}\label{equilibria_subsection}
%%%%%%%%%%%%%%%%%%%%

With the previous lemmas and theorem we have proved that $T(t)$ has an u.a.s. global attractor. Roughly speaking this means that, no matter the choice of initial conditions $\xvec(0)\in X$, $T(t)\xvec(0)$ will converge asymptotically to a compact bounded invariant set $A$. Since in our case $X$ is a compact subset of $\Rbb^{2n_{M}+2n_{DCG}}$, $A$ can contain only equilibrium points, periodic orbits and strange attractors, and all of them are asymptotically stable~\cite{hale_2010_asymptotic}. We 
first show that the dynamics converge exponentially fast to the equilibria. We will then argue about the absence of periodic orbits and strange attractors. 

First of all, it can be trivially seen from sections \ref{eq3_sec} and \ref{eq4_sec} that equilibria must satisfy $u_j^Tv_{M}+v_{0_j}=-v_{c},0,v_{c}$ for any $j$. Moreover, as discussed in section \ref{MEM_SOLC_subsection} $u_j^Tv_{M}+v_{0_j}=0$ leads to unstable equilibria, while $u_j^Tv_{M}+v_{0_j}=\pm v_{c}$ leads to asymptotically stable ones. However, these equilibria can be reached if the necessary condition 
\begin{equation}
\frac{1}{2}+\frac{\sqrt{3}}{6}<s<s_{\max}\text{ \ and \ }|i_{DCG_j}|<|i_{DCG_{\max}}| \label{condition_s_iDCG_stable}
\end{equation}
holds (see section \ref{eq3_sec} and \ref{eq4_sec}). It is also easy to see from Secs.~\ref{eq3_sec} and \ref{eq4_sec} that, by construction, this is a necessary condition to have equilibrium points for $T(t)$. However, this does not guarantee that at equilibrium $i_{DCG}=0$. 

For our purposes, we need to have $i_{DCG}=0$. In fact, at the equilibria the SOLC has voltages at gate nodes that can be only either $-v_{c}$ or $v_{c}$. In such configuration of voltages, as discussed in section \ref{MEM_SOLG_subsection}, the gates can stay in correct or non-correct configuration. In the former case no current flows from gate terminals due to the DCMs and so the correspondent component of $i_{DCG}$ must be $0$ at the equilibrium. On the other hand, if the gate configuration is not correct, at equilibrium we have currents of the order of $v_c/R_{on}$ that flow from the gates terminals (Sec.~\ref{MEM_SOLG_subsection}). These currents can be compensated only by the components of $i_{DCG}$. Therefore, if we indicate with $K_{wrong}v_c/R_{on}$ the minimum absolute value of the current flowing from the terminals of the gates when in the wrong configuration, 
$K_{wrong}=O(1)$, and consider VCDCG with
\begin{equation}
i_{DCG_{\max}}<K_{wrong}v_c/R_{on} \label{condition_iDCG},
\end{equation}
we have that the equilibria with nonzero components of $i_{DCG}$ disappear and only equilibria for which $i_{DCG}=0$ survive. With this discussion we have then proven the following theorem

\begin{thm}
	\label{equilibria_theorem}If the condition \eqref{condition_iDCG} holds, the u.a.s. stable equilibria for $T(t)$, if they exist, satisfy
	\begin{align}
	& i_{DCG_j}   =0\label{equilibrium_cond_iDCG}\\
	& s_j   =s_{\max}\label{equilibrium_cond_s}\\
	& |u_j^Tv_{M}+v_{0_j}|   =v_{c} \label{equilibrium_cond_vc}%
	\end{align}
	for any $j=1,\cdots,n_{DCG}$. Moreover, this implies that the gate relations are all satisfied at the same time.
\end{thm}

This theorem is extremely important because it is the same as: 
\begin{thm}
$T(t)$ has equilibria iff the CB problem implemented in the SOLC has solutions for the given input.
\end{thm}

We can analyze the equilibria even further. In fact we can prove their exponential convergence. With this aim in mind, we first analyze what happens to Eq.~\eqref{memristor} when we are at an equilibrium. In this case, for each memristor we can have two possible cases: $v_{M_j}=-R_{off}|i_{M_j}|$, with $|i_{M_j}|$ the absolute value of the current flowing through the memristor and it is an integer $>1$ times $v_c/R_{off}$ (this can be proven substituting values of $v_o$,  $v_1$ and  $v_2$ in equation \eqref{VDVG} that satisfies the SO-gates and using coefficients of table \ref{table1}) and
then $x_j=1$. In the second case we have $v_{M_j}=0$ and $x_j$ can take any value in the range $[0,1]$. The latter case implies that the equilibrium is not unique for a given $v_{M}$ but we have a continuum of equilibria, all of them with the same $v_{M}$, $s$ and $i_{DCG}$ but different $x$. The indetermination of some components of $x$ (those related to the components $v_{M}$ equal to 0) creates center manifolds around the equilibria.
However, these center manifolds are irrelevant to the equilibria stability since they are directly related to indetermination of the components of $x$ and these components can take any value in their whole range $[0,1]$. Therefore, we have to consider only the stable manifolds of the equilibria. 

In conclusion, since in the stable manifolds $C^{r}$ semigroups with $r\geq1$ have {\it exponential convergence} \cite{perko_01}, and in our case the center manifolds do not affect the convergence rate, this proves the following theorem

\begin{thm}
	The equilibrium points of $T(t)$ have {\it exponentially} fast convergence in all their attraction basin.
\end{thm}

Finally, in order to check the scalability of SOLCs, we want to study how the convergence of the equilibria depends on their size. We then write down the Jacobian of the system around an equilibrium point. Following the conditions in the theorem \ref{equilibria_theorem} and the equations discussed in the sections \ref{eq1_sec}-\ref{eq4_sec} it can be shown that the Jacobian of $F(\xvec)$ of equation \eqref{syst} evaluated in a equilibrium reads
\begin{equation}
J_{F}(\xvec=\xvec_{s})=\left(
\begin{array}
[c]{cccc}%
\frac{\partial F_{1}}{\partial v_{M}} & \frac{\partial F_{1}}{\partial x} &
C^{-1}A_{i} & 0\\
0 & \frac{\partial F_{2}}{\partial x} & 0 & 0\\
\frac{\partial f_{DCG}}{\partial v_{M}} & 0 & 0 & 0\\
0 & 0 & 0 & \frac{\partial f_{s}(i_{DCG},s)}{\partial s}%
\end{array}
\right)  \label{jacobian}.
\end{equation}
We also assume that in the second block row we have eliminated all equations for which $v_{M_j}=0$ holds, and from the second block column we have eliminated all columns related to the indeterminate $x_j$. This elimination is safe for our analysis since we want to study the eigenvalues of $J_{F}$. In fact, we notice that the eigenvectors related to the non-null eigenvalues are vectors with null components corresponding to the indeterminate $x_j$ since they are related to zero rows of $J_{F}$.

We can see from \eqref{jacobian} that, since the last block column and row of $J_{F}$ have all zero blocks but the last one, the eigenvalues of $J_{F}(\xvec=\xvec_s)$ are simply the union of the eigenvalues of $\frac{\partial f_{s}(i_{DCG},s)}{\partial s}$ and the eigenvalues of
\begin{equation}
J_{F_{red}}(\xvec=\xvec_s)=\left(
\begin{array}[c]{ccc}
\frac{\partial F_{1}}{\partial v_{M}} & \frac{\partial F_1}{\partial x} & C^{-1}A_{i}\\
0 & \frac{\partial F_{2}}{\partial x} & 0\\
\frac{\partial f_{DCG}}{\partial v_{M}} & 0 & 0
\end{array}
\right)  \label{jacobian_reduced}.
\end{equation}
Now, since $\frac{\partial f_s(i_{DCG},s)}{\partial s}$ is a diagonal matrix proportional to the identity $I$, more explicitly $\frac{\partial f_s(i_{DCG},s)}{\partial s}=-k_s(6s_{\max}^2-6s_{\max}+1)I$, its associated eigenvalues do not depend on the size of the circuit.

In order to study the spectrum of $J_{F_{red}}$ we notice that from Sec.~\ref{eq3_sec} we have $\frac{\partial f_{DCG}}{\partial v_{M}}=L_{DCG}\mathcal{U}$ where the derivative is evaluated in either $v_{DCG_j}=v_{c}$ or $-v_{c}$ according to the equilibrium point. So, the eigenvalues of $J_{F_{red}}$ are the time constants of an RLC memristive network. 

While it is not easy to say something about the time constants of a general RLC network, in our case there are some considerations that can be made. The capacitances, inductances, resistances are all equal (or very close to each other if 
noise is added). Moreover, the network is ordered, in the sense that there is a nontrivial periodicity and the number of connection per node is bounded and independent of the size. From these considerations, our network can actually be studied through its minimal cell, namely the minimal sub-network that is repeated to form the entire network we consider. This implies that the slower time constant of the network is {\it at most}
the number of cells in the network times the time constant of the single cell. Under these conditions we have then proved 
the following theorem:
\begin{thm}
	\label{poly-growth} Polynomially growing SOLCs support {\it at most} polynomially growing time constants.

\end{thm}
 
%%%%%%%%%%%%%%%%%%%%
\subsection{On the absence of periodic orbits and strange attractors}\label{absence_strange_subsection}
%%%%%%%%%%%%%%%%%%%%

In the previous sections we have proved that $T(t)$ is endowed with an u.a.s. global attractor. We have also provided an analysis of the equilibria proving their exponential convergence in the whole stable manifolds and discussed their convergence rate as a function of the size of the system, showing that this is at most polynomial. Therefore, in order to have a complete picture of a DMM physically realized with SOLCs, the last feature should be discussed: what is the composition of the global attractor.

In order to analyze the global attractor we use a statistical approach. We make the following assumptions:
\begin{enumerate}
	\item The capacitance $C$ is small enough such that, if we perturb a potential in a node of the network the perturbation is propagated in a time $\tau_{C}\ll\tau_{M}(\alpha)$ where $\tau_{M}(\alpha)$ is the switching time of the memristors (obviously linear functions of $\alpha$). For our system the time constant $\tau_{C}$ is related to the memristance and $C$.	
	\item $q$ of the function $f_{DCG}$ (see Fig.~\ref{DMM_figure_7}) is small enough such that the time $\tau_{DCG}(q)=i_{DCG_{\max}}/q$, satisfies $\gamma^{-1}\ll\tau_{DCG}$, i.e., the the time that the current $i_{DCG}$ takes to reach $i_{DCG_{\min}}$ is much smaller than the time it takes to reach $i_{DCG_{\max}}$.	
	\item The switching time of the memristors satisfies $\gamma^{-1}\ll\tau_{M}(\alpha)\ll\tau_{DCG}$.	
	\item The initial condition of $\xvec$ is taken randomly in $X$.
\end{enumerate}

Before proceeding we describe a peculiar behavior of our SOLCs that can be proved by the nonlocality induced by Kirchhoff's current laws and looking at how the DCMs work. If we change a potential of order $\delta V$ in a point of the network "quasi instantaneously", namely within a switching time $\tau_{C}\ll\tau\ll\tau_{M}(\alpha)$, there will be a sub-network $S$ - including the point switched - that in a time of the order of $\tau_{C}$ will change of the
same order of $\delta V$ many of the potentials at the nodes. This change is simply the result of the $RC$ nature of the network, the only component that survives at the time scale of $\tau_{C}$. After a time of the order of $\tau_{M}(\alpha)$, the sub-network $S$ will reach a stable configuration as consequence of the DCMs. Therefore, for a perturbation $\delta V$ in a node of the circuit, there is, in a time $\tau_{M}(\alpha)$, a
reconfiguration of an {\it entire region} $S$ of the network.

Now, from section \ref{eq3_sec} and \ref{eq4_sec}, we know that, if in a given node $j$ the current $i_{DCG_j}$ reaches $i_{DCG_{\max}}$ the equation that governs $i_{DCG}$ becomes $\frac{d}{dt}i_{DCG}=-\gamma i_{DCG}$, so the currents decrease to $i_{\min}$ in a time of the order of $\gamma^{-1}$. If we set $i_{\min}$ very small compared to the
characteristic currents of the network, in a time of the order of $\gamma^{-1}$ the potential at the node $j$ will experience a potential variation $\delta V $of the order of $R_{on}i_{\max}$, if $i_{\max}$ is large enough. 

This last claim can be proved because, if the current $i_{\max}$ is of the order of $K_{wrong}v_c/R_{on}$ and satisfies \eqref{condition_iDCG}, it means that the current from the network that compensates $i_{DCG_j}$ comes from some memristors that have been switched to $R_{on}$. Otherwise the current would be too small to reach values of the order of $K_{wrong}v_{c}/R_{on}$. Therefore, for each VCDCG that reaches $i_{\max}$ a voltage variation of the order of $R_{on}i_{\max}$ is produced at the node where the VCDCG is connected. Morever, since $\tau_{M}(\alpha)\ll\tau_{DCG}$, the network reconfigures the region $S$ before the
current $i_{DCG_j}$ starts increasing again.

With this in mind, and using the above conditions 1-4 we can make a statistical interpretation of our network. We consider a system large enough and, since it is made of a combination of elementary cells, we assume that $\tau_{C}$ is small enough such that the density of the nodes is uniform in the subnetwork $S$. Taking as initial $\xvec(0)$ a random point in $X$, we
have that at the rate of $\tau_{DCG}^{-1}$ a fraction $m_{DCG}/n_{DCG}$ of VCDCG reaches $i_{\max}$ and consequently switches. This means that there is a fraction of nodes that are kicked by $\delta V\approx K_{wrong}v_{c}/R_{on}$ at a rate of $\tau_{DCG}^{-1}$. 

Following our discussion above, if $\xvec_{-}$ is the configuration of the system 
before the kick, and $\xvec_{+}$ after (their temporal distance is of the order of $\gamma^{-1}$), then we have that the distance $\textrm{dist}(\xvec_-,\xvec_+)$ is of the order of the radius of $X$ defined as $\inf_{\mathbf{y}\in X}\sup_{\xvec\in X}\textrm{dist}(\xvec,\yvec)$. This means that these kicks bring the system in points of $X$ that are far from each other. Since we have chosen the initial $\xvec(0)$ random in $X$, also the kicks will take place randomly in time and in space (i.e., in the network node). This means that the system explores the entire $X$. It is worth noticing that, like in the case of Monte Carlo simulations, when we estimate an integral in $N$ dimensions \cite{Rossi_book}, here the SOLC explores $X$ in a way that is {\it linear} in the number of variables, i.e., it needs a number of kicks that grows at most linearly with the dimension of the system.

All this analysis allows us to conclude that, in our SOLCs, periodic orbits or strange attractors {\it cannot co-exist} with equilibria. In fact, both the periodic orbits and/or strange attractors, if they exist, produce large fluctuations of the potentials of the SOLC. These fluctuations are of the order of $v_{c}/R_{on}$ and are not localized but rather distributed in the entire network, because of the combination of VCDCGs and DCMs. Therefore, from the previous analysis, if periodic orbits and/or strange attractors exist they should force the system to explore the entire space $X$. However, if an equilibrium point exists, then, by exploring $X$ the system will intercept, 
within a time of order $\tau_{DCG}$ times a number that grows only linearly with the size of the system, the stable manifold of $X$ and collapses in the equilibrium point. Therefore, the global attractor is eiher formed by only equilibrium points or only by periodic orbits and/or strange attractors.

\subsection{Polynomial energy expenditure} We finally note that the SOLCs grow polynomially with the input size, each node of each gate can support only finite voltages (cut-off by the capacitances $C$) and at the equilibrium the voltages do not depend on the size of the system. The currents are also limited and their bounds are independent of the size of the SOLCs. The solution is found in a finite time which is polynomially dependent on the size of the SOLCs. Therefore, the energy expenditure can only grow polynomially with the SOLC size. 

This can be also seen in a more mathematical way as follows. The metric space $X=[v_{\min},v_{\max}]^{n_M}\times\lbrack0,1]^{n_M}\times[-i_{\max},i_{\max}]^{n_{DCG}%
}\times\lbrack s_{\min},s_{\max}]^{n_{DCG}}$ is a bounded compact space with the support of its range that does not depend on the SOLC size, hence, as a consequence the energy expenditure can only grow polynomially with SOLC size. 

We now provide numerical evidence of all these statements by solving two different NP problems, one hard.

%%%%%%%%%%%%%%%%%%%%
\section{NP problems solution with DMMs}\label{NP_sec}
%%%%%%%%%%%%%%%%%%%%

In this section we discuss two different SOLCs for the solution of two famous NP problems: the prime number factorization and the NP-hard version of the subset-sum problem. We show how, using the SO gates defined in  Sec.~\ref{MEM_SOLG_subsection}, and appropriate SOLCs, they can be solved with only polynomial resources. 

%%%%%%%%%%%%%%%%%%%%
\subsection{Prime Factorization}\label{factorization}
%%%%%%%%%%%%%%%%%%%%
\begin{figure}	
	\centerline{\includegraphics[width=1\columnwidth]{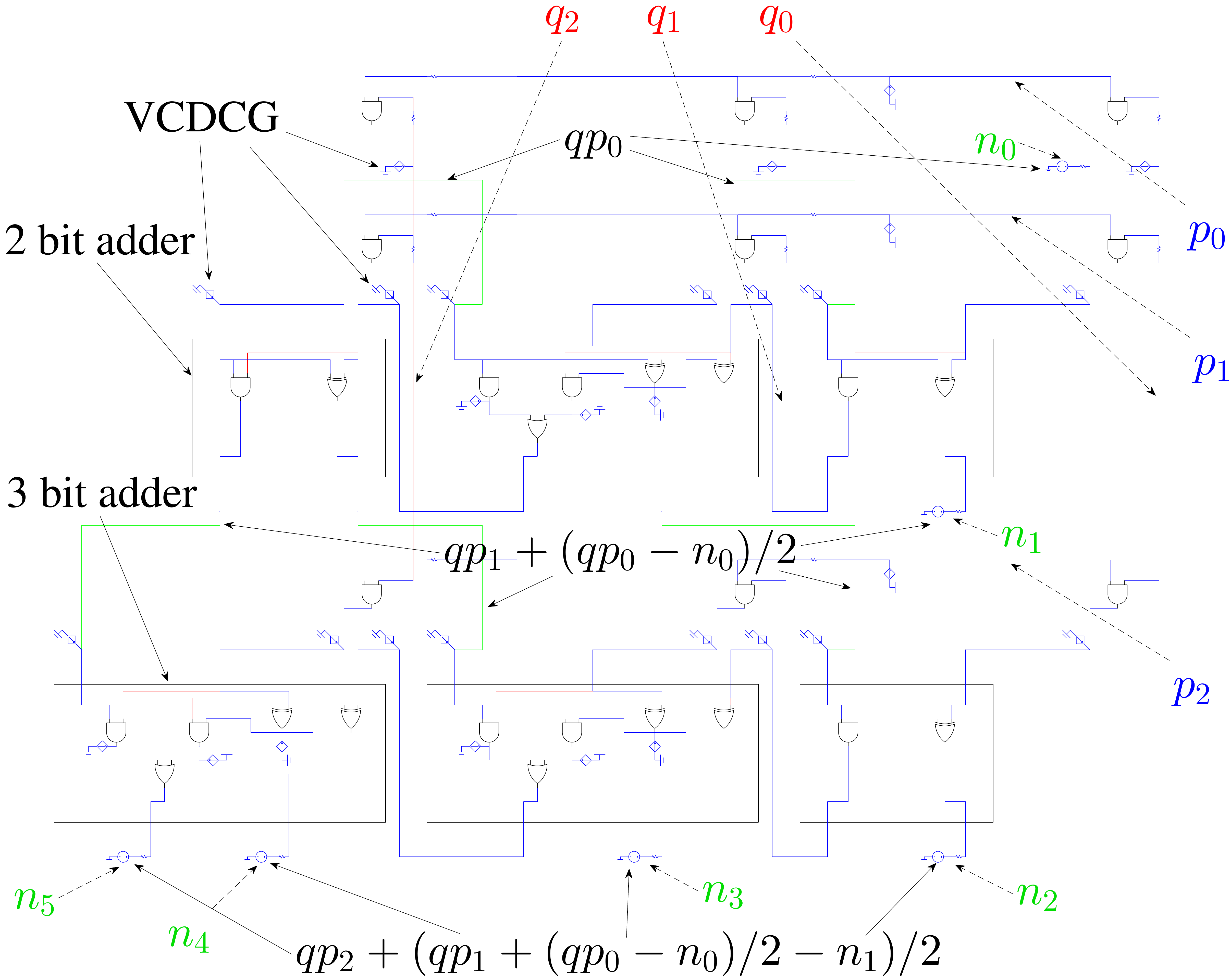}}			
	\caption{\label{DMM_figure_11}SOLC for solving a 6-bit factorization problem. The circuit is composed of the SOLGs described in section \ref{SOLG_section}.}	
\end{figure}
\begin{figure}	
	\centerline{\includegraphics[width=1\columnwidth,height=2.02\columnwidth,keepaspectratio]{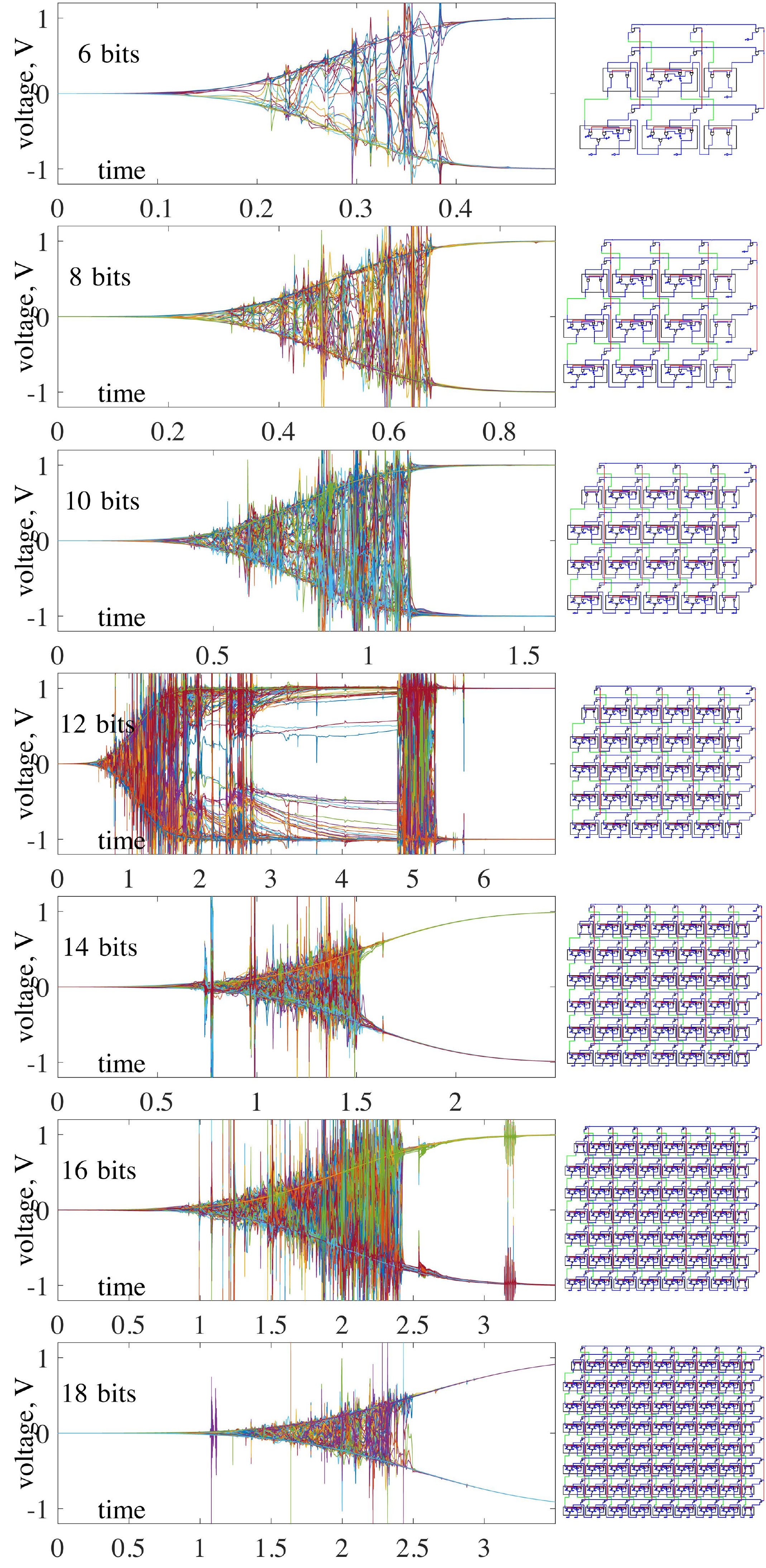}}			
	\caption{\label{DMM_figure_12}Voltages at the nodes of the SOLCs for the prime factorization as a function of time (in arbitrary units). The circuits simulated are reported on the right of each plot. They share the same topology as the one 
	in Fig.~\ref{DMM_figure_11}. The solution is found when all voltages are either 1 or -1 (logical 1 and 0, respectively).}	
\end{figure}

\begin{figure}	
	\centerline{\includegraphics[width=1\columnwidth]{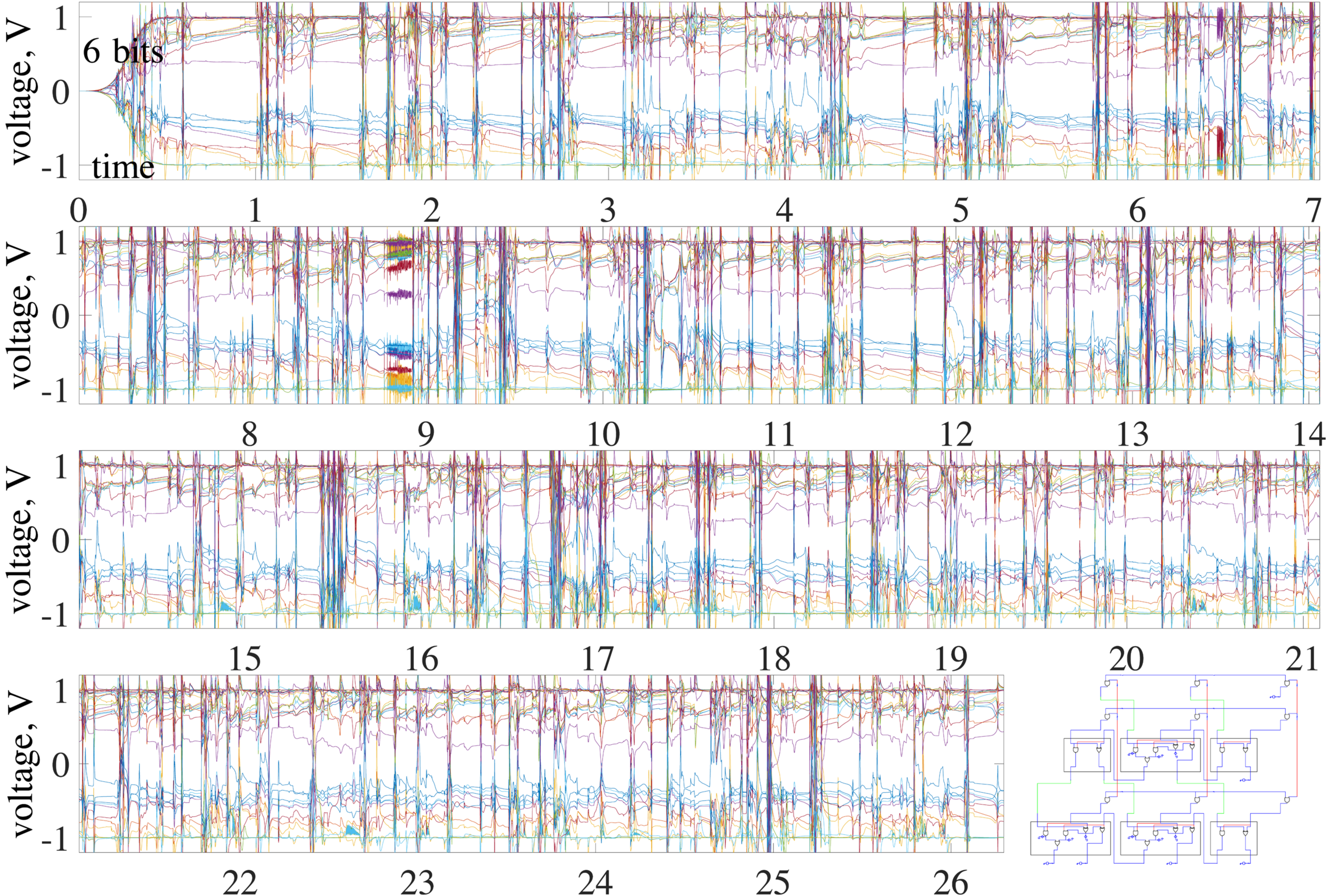}}			
	\caption{\label{DMM_figure_13}Voltages at the nodes of the SOLCs for the prime factorization as function of time (in arbitrary units). The circuits simulated is the 6-bit reported on the far bottom right of the plot. The input of the circuit is the prime 47.}	
\end{figure}

Let us consider an integer $n=pq$ where $p$ and $q$ are two prime numbers. Moreover, we define $p_j$, $q_j$, and $n_j$ the binary coefficients of $p$, $q$ and $n$ such that $p=\sum_{j=0}^{n_p} p_j 2^j$ and similarly for $q$ and $n$. Therefore, $n_p$, $n_q$ and $n_n$ are the number of bits of $p$, $q$, and $n$ respectively. The NP problem consists in finding the two unique primes $p$ and $q$ such that $n=pq$. 

In order to find $p$ and $q$ we can write the product of two numbers expressed in binary basis using a closed set of boolean functions. So we can express the product as $f(\yvec)=\bvec$ where $\bvec$ is the collection of the coefficients $n_j$, and $\yvec$ is the collection of the $p_j$ and $q_j$. It is easy to prove that $f$ is not unique, but for our purposes this is not relevant. 

According to our problem classification in Sec.~\ref{def_problem} the factorization problem belongs to the class CB, and we show here that it belongs also to MPI$_\text{M}$. In fact, starting from the $f(\yvec)=\bvec$ we can build the SOLC as reported in Fig.~\ref{DMM_figure_11} for $n_n=6$. The inputs of the SOLC are the generators indicated by the $n_j$. These generators impose voltages $v_c$ or $-v_c$ according to the logical value of $n_j$. Therefore, this set of generators is the control unit of the DMM and encodes $\bvec$. 

The lines at the same  potential indicated by $p_j$ and $q_j$ are the output of our problem, i.e., they encode - through the voltages - the coefficients $p_j$ and $q_j$ (namely $\yvec$). In order to read the output of the SOLC it is thus sufficient to measure the potentials at these lines. It is worth noticing that, once the gates have self-organized, the values of the potentials at all the lines will be either $v_c$ or $-v_c$. Therefore, {\it there is no issue with precision in reading the output}, implying robustness and scalability of the circuit. In Fig.~\ref{DMM_figure_11} the "intermediate steps" of SOLC implementation are indicated for the reader who is not familiar with boolean circuits.

{\it Scalability analysis -} It is clear that the circuit in Fig.~\ref{DMM_figure_11} can be built for any size of $\bvec$ and the number of SOLGs grows as $n_n^2$, so the SOLGs scale quadratically with the input. From the analysis in the previous sections, the equilibria of this SOLC are the solutions of the factorization problem and they can be exponentially reached in a time at most polynomial in $n_n$. Finally, since the energy of this circuit depends {\it linearly} on the time it takes to find the the equilibrium and on the number of gates, also the energy is bounded. We thus conclude that such circuits will solve factorization with polynomial resources if implemented in hardware. 

On the other hand, to simulate this circuit. i.e., to solve the ODE \eqref{syst}, we need a bounded small time increment $dt$ independent of the size of the circuit, and dependent only on the fastest time constant that can be associated to the time scales discussed in Sec.~\ref{absence_strange_subsection}. Therefore, if a solution exists to the prime factorization, and the SOLC fulfills the requirements of Sec.~\ref{equilibrium_subsection}, the problem belongs to the class MPI$_\text{M}$. More details on the complexity of this problem are discussed in section \ref{NP_P}.

It is worth noticing that the problem has no solution within this SOLC if either $n$ is already a prime, or at least one of $n_p$ or $n_q$ used to build the SOLC is smaller than the actual length of $p$ or $q$ used to solve the factorization problem\footnote{If the integer $n$ requires three or more primes to factorize, either the circuit needs to be extended to include extra output numbers, or we let the circuit with only $p$ and $q$ to break $n$ into two integers (depending on the initial conditions). We can then proceed by reducing 
even further these numbers till primes are found.}. This last scenario can always be avoided by simply choosing $n_p=n_n-1$ and $n_q=\lfloor n_n/2\rfloor$ (or the reverse), where $\lfloor\cdot\rfloor$ stands for the floor operator, i.e., it rounds the argument to the nearest integer towards minus infinity. This choice also guarantees that, if the $p$ and $q$ are primes, the solution is unique, and the trivial solution $n=n\times 1$ is forbidden. 

We have implemented the circuit into the NOSTOS (NOnlinear circuit and SysTem Orbit Stability) simulator \cite{Traversa_AEU,Traversa_IET,Traversa_TCAD,DCRAM,13_amoeba} developed by one of us (FLT). For the sake of simplicity, we have implemented SOLCs with $n_p=n_q$ with length at most of $n_n=n_p+n_q$. In this case, because of symmetry the possible solutions are two. 

{\it Numerical simulations -} In Fig.~\ref{DMM_figure_12} we present numerical results for several $n_n$. The simulations are performed by starting from a random configuration of the memristor internal variable $x$ and switching on gradually the generators. Although not necessary, we used a switching time for the generators which is quadratic in $n_n$. The plots report the voltages at the terminals of all the gates. Since we choose $v_c=1$, it can be seen that after a transient, all terminals approach $\pm v_c$, which are the logical 1 and 0. When, thanks to the DCMs, all of them converge to $\pm v_c$, they are necessarily satisfying all gate relations, and the SOLC has then found the solution. 

We have performed hundreds of simulations using a 72-CPU cluster, and have not found a single case in which the SOLCs did not converge to the equilibria, thus 
reinforcing the analysis in Sec.~\ref{absence_strange_subsection}. It is also worth noticing that the larger case we dealt with (18-bit case) requires the simulation of a dynamical system with approximatively 11,000 independent dynamic variables (i.e., $v_M$, $x$, $i_{DCG}$ and $s$). We are thus dealing with an enormous phase space and yet we did not find anything other than equilibria. Clearly, this does not prove the absence of strange attractors or limit cycles for all possible sizes, but at least for the parameters we have used (see table \ref{table2}) and the SOLC sizes we have tested this is the case. 

\begin{table}
	\caption{Simulation Parameters}
	\label{table2}
	\begin{center}
		\begin{tabular}{|l|c||l|c||l|c|}
			\hline 
			parameter  & value     & parameter  & value     & parameter  & value      \\ \hline 
			$R_{on}$   & $10^{-2}$ & $R_{off}$  & $1$       & $v_c$      & $1$        \\ \hline 
			$\alpha$   & $60$      & $C$        & $10^{-9}$    & $k$        & $\infty$   \\ \hline 
			$V_t$      & $0$       & $\gamma$   & $60$      & $q$        & $10$       \\ \hline 
			$m_0$      & $-400$    & $m_1$      & $400$     & $i_{\min}$ & $10^{-8}$  \\ \hline 
			$i_{\max}$ & $20$      & $k_i$      & $10^{-7}$ & $k_s$      & $10^{-7}$  \\ \hline 
			$\delta_s$ & $0$       & $\delta_i$ & $0$       &            &            \\ \hline 
		\end{tabular}
	\end{center}
\end{table}

Finally, in Fig.~\ref{DMM_figure_13}, we show the dynamics of the SOLC when we try to factorize a prime number. As expected, in this case the trajectories never find an equilibrium. Since the SOLC has nonetheless a global attractor, what we are seeing in this plot could be either a complicated periodic orbit or a strange attractor. 

%%%%%%%%%%%%%%%%%%%%
\subsection{The Subset-Sum Problem}\label{SSPSec}
%%%%%%%%%%%%%%%%%%%%
\begin{figure}	
	\centerline{\includegraphics[width=1\columnwidth]{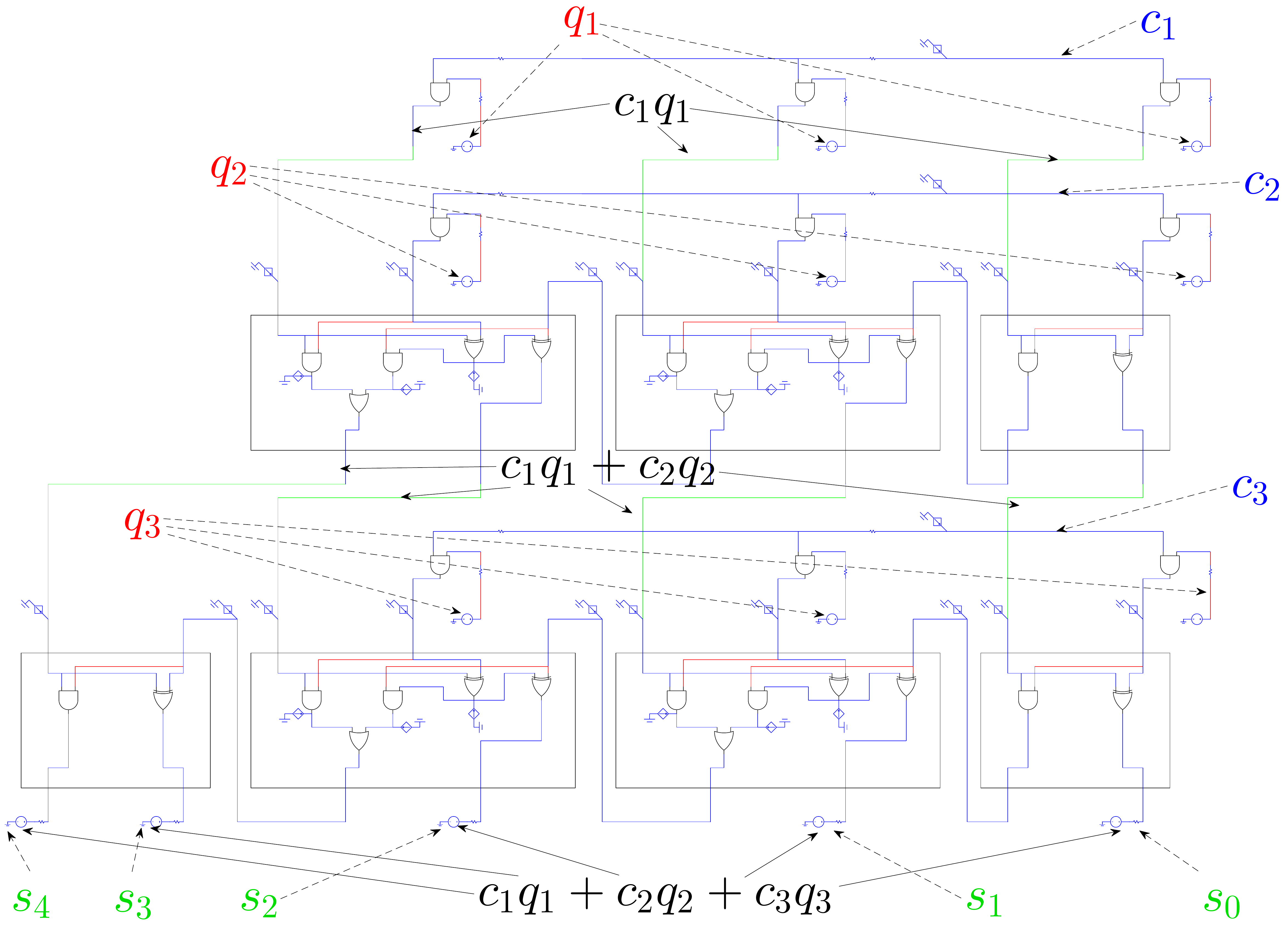}}			
	\caption{\label{DMM_figure_14}SOLC for solving a 3-number, 3-bit subset-sum problem. The circuit is composed of the SOLGs described in section \ref{SOLG_section}.}	
\end{figure}
\begin{figure}	
	\centerline{\includegraphics[width=1\columnwidth,height=2.02\columnwidth,keepaspectratio]{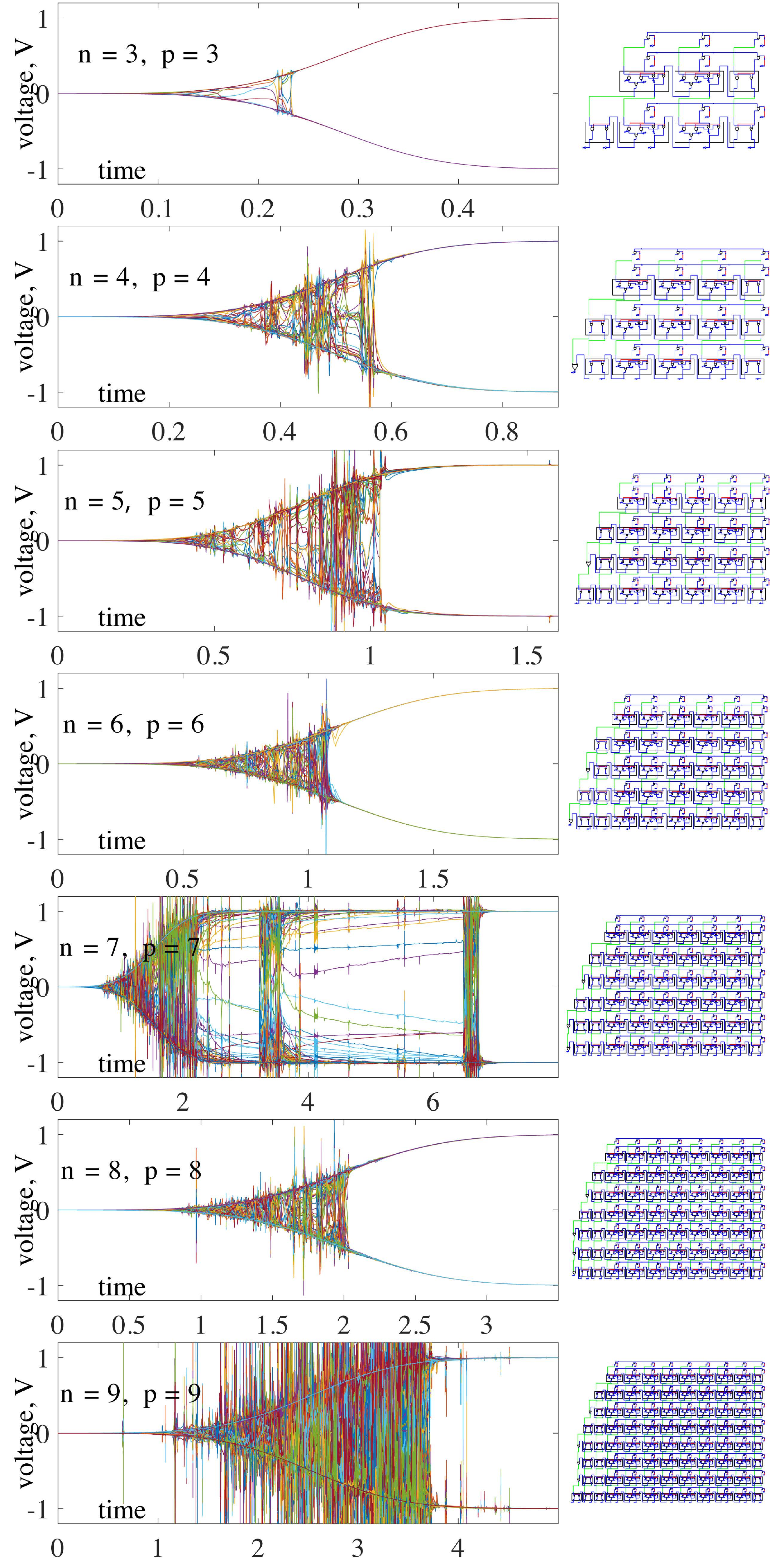}}			
	\caption{\label{DMM_figure_15}Voltages at the nodes of the SOLCs for the SSP in function of time in arbitrary units. The circuits simulated are reported on the right of each plot.}	
\end{figure}

We now show how to solve, with polynomial resources, within a SOLC the NP--hard version of the \textit{subset-sum problem} (SSP), which is arguably one of the most important problems in complexity theory \cite{complexity_bible}. The problem is defined as follows: if we consider a finite set $G\subset\mathbb{Z}$ of cardinality $n$, we want to know whether there is a non-empty subset $K\subseteq G$ whose sum is a given number $s$ (NP-complete version). In addition, we will seek to find at least one subset (if it exists) that solves the problem (NP-hard version). The complexity of the SSP is formulated in terms of both its cardinality ($n$) and the minimum number of bits used to represent any element of $G$ (namely the precision $p$). The problem becomes difficult to solve when $n$ and $p$ are of the same order because the known algorithms to solve it are exponential in either $n$ or $p$ \cite{traversaNP,complexity_bible}.

We consider here the NP-hard version of the SSP in which all elements of $G$ are positive. The case in which they are not can be equally implemented in the SOLCs but requires a slightly more complex topology and will be reported elsewhere. In order to solve the SSP we try to find the collection of $c_j\in \Zbb$ such that 
\begin{equation}
\sum_jc_jq_j=s\label{SSP_equation}
\end{equation}
where $q_j\in G$ and $j=1,\cdots,n$. Therefore, our unknowns are the $c_j$, with $\yvec$ the collection of the $c_j$. Equation~\eqref{SSP_equation} can be readily represented in boolean form through a boolean system $f(\yvec)=\bvec$ where $\bvec$ is composed by the binary coefficients of $s$ padded with a number of zeros such that $\dim(\bvec)$ equals the minimum number of binary coefficients used to express $\sum_jq_j$. It is easy to show that $\dim(\bvec)\leq \log_2(n-1)+p$. This boolean system can be implement in a SOLC as shown in Fig.~\ref{DMM_figure_14}. The control unit is here composed by the generators that implement the binary coefficients of $s$ and the output can be read out by measuring the voltages at the lines pointed out by the $c_j$ in Fig.~\ref{DMM_figure_14}. 

{\it Scalability analysis -} This circuit grows {\it linearly} in $p$ and as $n+\log_2(n-1)$ in $n$. The last term is due to the successive remainders during the sum in \eqref{SSP_equation}. This is represented by the extra adders on the left of the SOLCs in Figs.~\ref{DMM_figure_14} and \ref{DMM_figure_15}. Also in this case, like the factorization, we have a SOLC that grows polynomially with both $p$ and $q$ so the SSP belongs to the MPI$_\text{M}$ class (see Sec.~\ref{def_problem}), indicating that, unlike the Turing paradigm, factorization and the subset-sum problem share the 
same complexity within the memcomputing paradigm.

{\it Numerical simulations -} We have performed simulations of the SOLCs in a similar way as the factorization. The results are reported in Fig.~\ref{DMM_figure_15}. As in the case of 
the factorization, the solution is found when all gate voltages are either 1 or -1 (logical 1 and 0, respectively). We have performed an extensive analysis of this case as well and 
up to $n=9$ and $p=9$ no periodic orbits or strange attractors have been found by starting from random initial conditions. Finally, also for the subset-sum problem, if no solution exists the system will not converge to any equilibrium. 

%%%%%%%%%%%%%%%%%%%%
\section{Discussion on the NP=P problem}\label{NP_P}
%%%%%%%%%%%%%%%%%%%%

In this work we have introduced and proved many statements derived from the new concept of Memcomputing Machines, and in particular the easily scalable subclass we have named Digital Memcomputing Machines, and their ability to solve complex problems with polynomial resources. In this section we discuss the main results we have obtained in relation to the problem of whether NP=P or not.

The first result related to this problem has been provided in section \ref{equilibrium_subsection}. There, we have summarized the practical mathematical constraints that a DMM must satisfy in order to solve an NP problem employing only polynomial resources. It is worth noticing that the constraints we found are compatible with each other. Therefore, the conclusion is that a dynamical system that solves NP problems with polynomial resources can exist. However, this does not solve the NP=P problem yet, since the existence of such a system is not enough:  finding it may require exponential resources. 

We have then tried to design a system that satisfies all constraints given in Sec.~\ref{equilibrium_subsection}. We rigorously proved in Sec.~\ref{Math_SOLC_section} that all of them are satisfied except for the absence (or irrelevance) of limit cycles and strange attractors for which we gave arguments in Sec.~\ref{absence_strange_subsection}. Therefore, the NP=P problem has not been settled yet, but the 
present work provides strong support to the answer that indeed NP=P. This comes from the polynomial scaling of the simulations shown in Sec.~\ref{NP_sec} and the fact that we have found no limit cycles or strange attractors when equilibria were present.

The resources used to simulate DMMs and, in particular, SOLCs can be quantified in the number of floating point operations the CPU does to simulate them. Since we are actually integrating an ODE, $\dot\xvec=F(\xvec)$,  (Eq.~\eqref{syst}), the number of floating point operations depends \textit{i}) {\it linearly} (if we use a forward integration method like the forward Euler or Runge-Kutta~\cite{Stoer}) or at most {\it quadratically} (if we use a backward method like backward Euler or Trapezoidal rule~\cite{Stoer}) in $\dim(\xvec)$, and \textit{ii}) depends on the  minimum time increment we use to integrate with respect to the total period of simulation, or in other words depends {\it linearly} on the number of time steps $N_{t}$. We discuss them separately.

We have seen in Sec.~\ref{NP_sec} that for NP problems we have that the $\dim(\xvec)$ scales polynomially in the input of the problem (quadratically in the number of bits for the factorization, and linearly in both the precision and the cardinality of the set $G$ for the SSP). Note also that we solve these NP problems by mapping them into a more general NP-complete problem, the Boolean satisfiability (SAT) problem \cite{complexity_bible} and then we build the SOLCs by encoding directly the SAT representing the specific problem (this SAT is in compact boolean form that we indicate with $f(\yvec)=\bvec$). This means that the $\dim(\xvec)$ depends {\it linearly} on the number of elementary logic gates (i.e., AND, OR, XOR) used to represent the SAT.

The number of time steps, $N_t$, to perform the simulations has a double bound. The first one depends on the minimum time increment $\Delta t$, and the second on the minimum period of simulation $T_s$. As discussed in \ref{factorization}, the former depends on the smallest time constant of the SOLC. Ultimately, this time constant {\it does not depend} on the size of the circuit, but on the nature of the single devices we are simulating. On the other hand, $T_s$ can depend on the size of the system. In fact, it is the minimum time we need to clearly find the equilibria of the system. It is thus related to the largest time constant of the system. However, in section \ref{Math_SOLC_section} we proved that $T_s$ grows {\it at most polynomially} 
with the size of the problem.

From all this, one can then infer that we can simulate {\it polynomially} a DMM using a Turing machine, suggesting NP=P. We stress again, though, that we only gave theoretical and numerical 
arguments that, if equilibria exist, the global attractor does not support periodic orbits and/or strange attractors. We could not come up with a formal proof of this last statement. In other words, although we have evidence that this could be the case, as discussed in Sec.~\ref{absence_strange_subsection}, and shown numerically in 
Sec.~\ref{NP_sec}, a rigorous formal proof is still lacking and we cannot settle the question of whether NP=P. (Note that it would be enough to show that the basin of attraction satisfies the statistical conditions discussed in section \ref{equilibrium_subsection}.)

%%%%%%%%%%%%%%%%%%%%
\section{Conclusions}\label{Conclusions}
%%%%%%%%%%%%%%%%%%%%
With this paper we have defined digital (hence scalable) memcomputing machines that use memory to store and process information and are able to solve complex problems (including NP problems) with polynomial resources. Taking advantage of dynamical system theory, we have provided the necessary mathematical constraints that the DMMs must satisfy for the poly-resource resolvability of exponentially difficult problems, and found that these constraints are compatible with each other.

We have proposed a practical implementation based on the concept of self-organizing logic gates and circuits that can solve  boolean problems by self-organizing into their solution. We have used tools of functional analysis to rigorously prove the requirement of polynomial resources with input size of these circuits, except for the co-existence of limit cycles and strange attractors with the equilibria of the problem. Therefore, we could not formally prove that NP=P. 

Using these SOLC realizations of DMMs we have solved prime factorization and the NP-hard version of the subset-sum problem using polynomial resources. The former problem scales as $O(n^2)$ in space (i.e., 
with the number of self-organizing logic gates employed) and $O(n^2)$ in convergence time with input size $n$. The latter as $O[p(n+\log_2(n-1))]$ in space and $O((n+p)^2)$ in convergence time with size $n$ and precision $p$. 

These machines are not just a theoretical concept. They can be fabricated either with circuit elements with memory (such as memristors) and/or standard MOS technology. They do not require cryogenic temperatures or to maintain quantum coherence  to work since they are essentially classical objects, and map integers into integers. Therefore, they are quite robust against noise like our present digital Turing machines implemented within the von Neumann architecture. 

There are several directions that one can take from here, both from the experimental and theoretical point of view. Experimentally, it would be of course desirable to build these machines in the lab. At this juncture, the choice of materials that would best fit their possible integration suggests an implementation using only MOS technology, rather than a combination 
of MOS devices and, e.g., memristors. By MOS technology we do not mean here the use of transistors in their switch mode only. Instead, we envision {\it emulators} of memristors that would perform the functions we have described in this paper. 

From the theory side, one is immediately drawn to consider new encryption protocols against these machines, in view of the fact that they can solve NP-hard problems in polynomial time. Encryption schemes not based on the existing complexity theory but rather on the complexity classes NPI$_\text{M}$ or NPC$_\text{M}$ we have introduced in Sec.~\ref{complexity_section} may be the way to proceed, but the options at this point are wide open. 

On a more speculative and possibly just academic scenario, we conclude by noting that one could extend the concept of DMMs to topological vector spaces, in particular Hilbert spaces, thus allowing an extension of our concepts to quantum systems. There is indeed some recent interest in fundamental memory elements such as superconducting memristors~\cite{Peotta} and quantum memristive elements~\cite{Pfeiffer} that take advantage of quantum phenomena. A combination of such elements 
in circuits with the appropriate topology may allow the solution of even more complex problems (e.g., super-exponential ones) in polynomial time. 

Irrespective, the machines we propose are a realistic alternative to Turing machines and could find applications in a wide variety of fields, such as machine learning, robotics, real-time computing, etc. 

\section{Acknowledgments}
We thank Hava Siegelmann, Fabrizio Bonani, and Bj\"{o}rn Tackmann for useful discussions. This work has been partially supported by the Center for Memory Recording Research at UCSD.

%\section*{References}
%\bibliographystyle{unsrt}
\bibliographystyle{ieeetr}%{IEEEtran}
\bibliography{UMM}

\end{document}